\documentclass[12pt]{article}
\usepackage{amssymb,amsmath,amsthm,amsfonts,mathrsfs}
\usepackage{enumerate}
\usepackage{ragged2e}
\usepackage{pgf,tikz}
\usetikzlibrary{intersections}

\usepackage{ucs} 
\usepackage[utf8x]{inputenc}

\usepackage{tikz}
\usepackage{tkz-tab}
\usepackage{caption}
\usepackage{latexsym}
\usepackage{amssymb}
\usepackage{amsmath}
\usepackage{subcaption} 
\usepackage{graphicx}
\usepackage{pgfplots}
\pgfplotsset{compat=1.15}
\usepackage{mathrsfs}
\usetikzlibrary{arrows}
\usepackage{float}
\usepackage{amsmath}
\usepackage[colorlinks=true, allcolors=blue]{hyperref}
\usepackage[all]{hypcap}
\usepackage[capitalize,nameinlink]{cleveref}
\usepackage{titlesec}
\titlelabel{\thetitle.\quad}
\usepackage{float}
\usepackage{pgfplots}
\pgfplotsset{compat=1.15}
\usepackage{mathrsfs}
\usetikzlibrary{arrows}

\newtheorem{theo}{Theorem}
\newtheorem{obs}[theo]{Observation}

\newtheorem{lem}[theo]{Lemma}

\theoremstyle{definition}
\newtheorem{case}{Case}[theo]

\usepackage{amsmath}

\newcommand*{\QEDA}{\hfill\ensuremath{\blacksquare}}
\usepackage{ragged2e}

\begin{document}
\title{ Representation and Obstruction characterization of co-TT graphs}
\def\correspondingauthor{\footnote{Corresponding author}}
\author{Ashok Kumar Das\correspondingauthor{}  and  Indrajit Paul\\Department of Pure Mathematics\\
University of Calcutta\\
35, Ballygunge Circular Road\\
Kolkata-700019\\
Email Address - ashokdas.cu@gmail.com \&
paulindrajit199822@gmail.com}
\maketitle
\begin{abstract}
  Threshold tolerance graphs and their complement graphs, known as co-TT graphs, were introduced by Monma, Reed, and Trotter\cite{24}. Building on this, Hell et al.\cite{19} introduced the concept of negative interval. Then they proceeded to define signed-interval digraphs/ bigraphs, demonstrating their equivalence to several seemingly distinct classes of digraphs/ bigraphs.  They also showed that co-TT graphs are equivalent to symmetric signed-interval digraphs, where some vertices of the digraphs have loops and others do not. We have showed that this actually solve the representation characterization problem of co-TT graphs posed by Monma, Reed and Trotter \cite{24}.

In this paper, we characterize signed-interval bigraphs and signed-interval graphs in terms of their biadjacency matrices and adjacency matrices, respectively. Moreover we emphasize on the geometric representation of signed-interval graphs, i.e. co-TT graphs. Finally, by utilizing the geometric representation of signed-interval graphs, we resolve the open problem of characterizing  co-TT graphs in terms of minimal forbidden induced subgraphs, a problem initially posed by Monma, Reed, and Trotter in the same paper.
\end{abstract}
\par\noindent \textbf{Keywords:} interval graphs, signed-interval graphs, negative interval, co-TT graphs, forbidden induced subgraphs.
\section{Introduction}
A graph $G = (V, E)$ is classified as an \textit{interval graph} if, for each vertex $v \in V$, we can associate an interval on the real number line in such a way that two vertices are connected by an edge if and only if their corresponding intervals intersect. Interval graphs hold significant importance in the field of graph theory due to their several elegant characterizations and the existence of efficient recognition algorithms \cite{1, 3, 4, 10, 12, 13, 18, 20, 23, 26}. Furthermore, they find practical applications in various real-world problems.\par

An interval digraph is a \textit{directed graph} representable by assigning each vertex $v$ an ordered pair $(S_v, T_v)$ of closed intervals so that $uv$ is a (directed) edge if and only if $S_u$ intersects $T_v$. The intervals $S_v$ and $T_v$ are the \textit{source interval} and the \textit{sink interval} for $v$. A bipartite graph (in short, bigraph) $B = (X, Y, E)$ is an interval bigraph if there exists a one to one correspondence between the vertex set $X \cup Y$ of $B$ and a collection of intervals $\{I_v: v \in X \cup Y\}$ so that $xy \in E$, if and only if $I_x \cap I_y \neq \emptyset$, where $x\in X$ and $y\in Y$.\par

The \textit{biadjacency matrix} of a bipartite graph is the submatrix of its adjacency matrix consisting of the rows indexed by the vertices of one partite set and columns by the vertices of another. Interval digraphs and interval bigraphs were initially introduced in \cite{27} and \cite{16}, respectively. As observed in \cite{5}, these two concepts are equivalent.\par

The key observation is that the adjacency matrix of an interval digraph can be transformed into the biadjacency matrix of an interval bigraph, and conversely, the biadjacency matrix of an interval bigraph can be converted into the adjacency matrix of an interval digraph. This transformation may require the addition of rows or columns of '0s' to ensure that the resulting matrix is square. \par
\indent\hspace{.05in}Next we switch over to interval graphs. Let $G= (V, E)$ be an interval graph and $\{I_v : v \in V\}$ be its interval representation. Assume $l(v)$ and $r(v)$ respectively denote the left end point and right end point of $I_v$. It can be observed that two intervals $I_u$ and $I_v$ intersect if and only if left end point of each interval is less than or equal to the right end point of the other. Thus $I_u \cap I_v \neq \emptyset$ if and only if $l(u) \leq r(v)$ and $l(v) \leq r(u)$.\par
Here $l(v)$ and $r(v)$ respectively denote the left end point and right end point of $I_v$. Again for $I_v$ to be an interval it is needed that $l(v) \leq r(v)$.\\
\par Based on this observation, Hell et al. \cite{19} proposed a generalization of the interval model by relaxing the constraint $l(v) \leq r(v)$. They introduced the concepts of positive and negative intervals as follows: An interval $I_v$ is classified as a \textit{positive interval} if $l(v) \leq r(v)$, and it is denoted by $I^+_v$. Conversely, $I_v$ is considered a \textit{negative interval} if $l(v) > r(v)$, and it is denoted by $I^-_v$.\par
The motivation of the concept of negative interval comes from the definition of complement of Threshold tolerance graph \cite{24}. A graph $G$ is a \textit{Threshold tolerance graph} if corresponding to each vertex $v$ we can assign a positive weight $w_v$ and a positive tolerance $t_v$ so that $uv$ is an edge if and only if $w_u + w_v$ $>$ $min\{t_u, t_v\}$.  Threshold tolerance graphs are threshold graphs when we assume all the tolerances to be equal. In other words threshold tolerance graphs are the generalization of threshold graphs. 
\par The complements of threshold tolerance graphs are called \textit{co-TT graphs}. This class of graphs not only contains all threshold graphs ( since complement of a threshold graph is a threshold graph) but also related to other well-studied classes of graphs (like interval graphs). The following characterization of co-TT graphs is due to Monma, Reed and Trotter\cite{24}.\\
\begin{theo}
A graph G is a co-TT graph if and only if we can assign positive numbers $a_v$ and $b_v$ for each $v \in V$ such that,
 $$xy \in E(G) \Leftrightarrow a_x \leq b_y \text{ and } a_y \leq b_x \hspace{1cm}\ldots (1)$$
\end{theo}
For a co-TT graph $G$ the above condition is equivalent to assignment of a pair of positive numbers $(a_v, b_v)$ to each vertex $v$ of $G$, where it is possible that $a_v \leq b_v$ or $a_v > b_v$. If $a_v \leq b_v$, we say that the pair of positive numbers $(a_v, b_v)$ forms a positive interval ${I_v}^+ = [a_v, b_v]$ and in the other case if $a_v > b_v$ then we say that the pair $(a_v, b_v)$ forms a negative interval ${I_v}^- = [a_v, b_v]$. It is easy to observe that when two intervals ${I_x}^+ = [a_x, b_x]$ and ${I_y}^+ = [a_y, b_y]$ intersect then the condition $(1)$ is satisfied. Also if one of the two intervals, say ${I_x}^-$ (or, ${I_y}^-$) is a negative interval and ${I_x}^- \subseteq {I_y}^+$ $(or {I_y}^- \subseteq {I_x}^+ )$ then also the condition $(1)$ is satisfied. Thus we are in a position to give a general definition of interval graphs which Hell et al. \cite{19} called signed- interval graphs. A graph $G = (V, E)$ is a \textit{signed-interval graph} if corresponding to each vertex $v$ we assign either positive interval ${I_v}^+$ or negative interval ${I_v}^-$ such that $uv \in V(E)$ if and only if ${I_u}^+ \cap {I_v}^+ \neq \emptyset$ or ${I_v}^- \subseteq {I_u}^+$ or ${I_u}^- \subseteq {I_v}^+$.\par

\begin{figure}[H]
    \centering
    \begin{tikzpicture}[line cap=round,line join=round,x=1.0cm,y=1.0cm, scale=.5]
\clip(8.39,3.2) rectangle (32.7,22.79);
\draw [line width=0.2pt] (18.,18.)-- (18.,16.);
\draw [line width=0.2pt] (18.,16.)-- (16.,14.);
\draw [line width=0.2pt] (16.,14.)-- (20.,14.);
\draw [line width=0.2pt] (18.,16.)-- (20.,14.);
\draw [line width=0.2pt] (20.,14.)-- (22.,12.);
\draw [line width=0.2pt] (16.,14.)-- (14.,12.);
\draw [line width=0.2pt] (12.017991936014672,7.192721458372999)-- (16.01780410187087,7.192721458372999);
\draw [line width=0.2pt] (10.018085853086575,6.421673330015169)-- (12.017991936014672,6.397578076003987);
\draw [line width=0.2pt] (14.01789801894277,6.397578076003987)-- (20.017616267727064,6.421673330015169);
\draw [line width=0.2pt] (16.017804101870873,5.602434693634976)-- (22.01752235065517,5.602434693634976);
\draw [line width=0.2pt] (22.01752235065517,6.445768584026352)-- (24.01742843358327,6.397578076003987);
\draw [->,line width=1.5pt] (18.017710184798975,4.7831960572547825) -- (16.017804101870876,4.7831960572547825);
\draw [line width=1.5pt] (16.017804101870876,4.7831960572547825)-- (14.017898018942777,4.807291311265964);
\draw (17.4,19.2) node[anchor=north west,scale=1.3] {$u$};
\draw (21.9,12.9) node[anchor=north west, scale=1.4] {$v$};
\draw (20.1,14.8) node[anchor=north west,scale=1.3] {$y$};
\draw (17.9,16.8) node[anchor=north west, scale=1.4] {$x$};
\draw (14.6,14.8) node[anchor=north west,scale=1.4] {$z$};
\draw (12.6,12.7) node[anchor=north west,scale=1.2] {$w$};
\draw (9.5,9) node[anchor=north west] {$1$};
\draw (11.5,9) node[anchor=north west] {$2$};
\draw (13.5,9) node[anchor=north west] {$3$};
\draw (15.5,9) node[anchor=north west] {$4$};
\draw (17.5,9) node[anchor=north west] {$5$};
\draw (19.5,9) node[anchor=north west] {$6$};
\draw (21.5,9) node[anchor=north west] {$7$};
\draw (23.5,9) node[anchor=north west] {$8$};
\draw (10.6,7.2) node[anchor=north west,scale=1.] {$u$};
\draw (13.093802789189901,8.0) node[anchor=north west,scale=1.] {$x$};
\draw (16.8,7.5) node[anchor=north west,scale=1.] {$y$};
\draw (22.6,7.3) node[anchor=north west,scale=1.] {$w$};
\draw (18.6,6.4) node[anchor=north west,scale=1.] {$z$};
\draw (15.1,5.7) node[anchor=north west,scale=1.] {$v$};
\draw (27.0,12) node[anchor=north west,scale=1.2] {$u: [1,2]$};
\draw (27.0,10.9) node[anchor=north west,scale=1.2] {$x: [2,4]$};
\draw (27.0,9.8) node[anchor=north west,scale=1.2] {$y: [3,6]$};
\draw (27.0,8.7) node[anchor=north west,scale=1.2] {$z: [4,7]$};
\draw (27.0,7.6) node[anchor=north west,scale=1.1] {$w: [7,8]$};
\draw (27.0,6.5) node[anchor=north west,scale=1.2] {$v: [5,3]$};
\begin{scriptsize}
\draw [fill=black] (18.,18.) circle (2pt);
\draw [fill=black] (18.,16.) circle (2pt);
\draw [fill=black] (16.,14.) circle (2pt);
\draw [fill=black] (20.,14.) circle (2pt);
\draw [fill=black] (22.,12.) circle (2pt);
\draw [fill=black] (14.,12.) circle (2pt);
\draw [fill=black] (10.,8.) circle (2pt);
\draw [fill=black] (12.,8.) circle (2pt);
\draw [fill=black] (14.,8.) circle (2pt);
\draw [fill=black] (16.,8.) circle (2pt);
\draw [fill=black] (18.,8.) circle (2pt);
\draw [fill=black] (20.,8.) circle (2pt);
\draw [fill=black] (22.,8.) circle (2pt);
\draw [fill=black] (24.,8.) circle (2pt);
\draw [fill=black] (12.017991936014672,7.192721458372999) circle (0pt);
\draw [fill=black] (16.01780410187087,7.192721458372999) circle (0pt);
\draw [fill=black] (10.018085853086575,6.421673330015169) circle (0pt);
\draw [fill=black] (12.017991936014672,6.397578076003987) circle (0pt);
\draw [fill=black] (14.01789801894277,6.397578076003987) circle (0pt);
\draw [fill=black] (20.017616267727064,6.421673330015169) circle (0pt);
\draw [fill=black] (16.017804101870873,5.602434693634976) circle (0pt);
\draw [fill=black] (22.01752235065517,5.602434693634976) circle (0pt);
\draw [fill=black] (22.01752235065517,6.445768584026352) circle (0pt);
\draw [fill=black] (24.01742843358327,6.397578076003987) circle (0pt);
\draw [fill=black] (18.017710184798975,4.7831960572547825) circle (0pt);
\draw [fill=black] (14.017898018942777,4.807291311265964) circle (0pt);
\end{scriptsize}
\end{tikzpicture}
    \caption*{Fig. 1. A co-TT graph and its co-TT representation; observe that $yv$ is an edge as $3\leq 3$ and $5\leq 6$ but $xv$ is not an edge as $5\nleq 4$.}
    \label{fig:enter-label}
\end{figure}
When all the intervals $I_v$ are positive in the representation, we obtain an interval graph. Additionally, based on the aforementioned observation, we can conclude that signed-interval graphs are equivalent to co-TT graphs. Hence, this also settle the problem of representation characterization for co-TT graphs posed by Monma, Reed and Trotter \cite{24}.\\
A graph $G$ is \textit{chordal} if it does not contain chordless cycle $C_n(n\geq 4)$. A vertex $v$ of $G$ is called a \textit{simplicial vertex} if its adjacent vertices induces a complete graph. Every chordal graph must have a simplicial vertex \cite{7}. Faber\cite{8} introduced a subclass of chordal graphs called strongly chordal graphs. A vertex $v$ of a chordal graph is a \textit{simple vertex} if the vertices in its closed neighbor are linearly ordered by inclusion of their closed neighborhood. A graph $G$ is a \textit{strongly chordal graph} if every induced subgraph of $G$ has a simple vertex. It is easy to see that every simple vertex is also a simplicial vertex of $G$.

\begin{figure}[H]
    \centering
    \begin{tikzpicture}[line cap=round,line join=round,x=1.0cm,y=1.0cm,scale=1.]
\clip(4,5) rectangle (12,11);
\draw [line width=0.2pt] (6.,8.)-- (10.,8.);
\draw [line width=0.2pt] (8.,10.)-- (6.,8.);
\draw [line width=0.2pt] (8.,10.)-- (10.,8.);
\draw [line width=0.2pt] (6.,8.)-- (6.,6.);
\draw [line width=0.2pt] (6.,6.)-- (10.,6.);
\draw [line width=0.2pt] (10.,8.)-- (10.,6.);
\draw [line width=0.2pt] (8.,8.)-- (8.,10.);
\draw [line width=0.2pt] (8.,8.)-- (8.,6.);
\draw [line width=0.2pt] (10.,8.)-- (8.,6.);
\draw [line width=0.2pt] (8.,8.)-- (10.,6.);
\draw [line width=0.2pt] (8.,8.)-- (6.,6.);
\draw [line width=0.2pt,] (6.,8.)-- (8.,6.);
\draw [shift={(6.,8.)},line width=0.2pt]  plot[domain=-0.7853981633974483:0.7853981633974483,variable=\t]({1.*2.8284271247461903*cos(\t r)+0.*2.8284271247461903*sin(\t r)},{0.*2.8284271247461903*cos(\t r)+1.*2.8284271247461903*sin(\t r)});
\draw (5.2,6.3) node[anchor=north west,scale=1.2] {$v_1$};
\draw (5.2,8.3) node[anchor=north west,scale=1.2] {$v_2$};
\draw (7.7,10.6) node[anchor=north west,scale=1.2] {$v_3$};
\draw (10.0,8.3) node[anchor=north west,scale=1.2] {$v_7$};
\draw (10.0,6.3) node[anchor=north west,scale=1.2] {$v_6$};
\draw (7.3,8.6) node[anchor=north west,scale=1.2] {$v_4$};
\draw (7.7,5.9) node[anchor=north west,scale=1.2] {$v_5$};
\begin{scriptsize}
\draw [fill=black] (6.,8.) circle (2.pt);
\draw [fill=black] (10.,8.) circle (2.pt);
\draw [fill=black] (8.,10.) circle (2.pt);
\draw [fill=black] (6.,6.) circle (2.pt);
\draw [fill=black] (10.,6.) circle (2.pt);
\draw [fill=black] (8.,8.) circle (2.pt);
\draw [fill=black] (8.,6.) circle (2.pt);
\end{scriptsize}
\end{tikzpicture}
    \caption*{Fig. 2. A strongly chordal graph and $v_1$ is a simple vertex, since $N[v_1]=\{ v_1, v_2, v_4, v_5\}$ and $N[v_1]\subseteq N[v_2]\subseteq N[v_4]= N[v_5]$.}
    \label{fig:enter-label}
\end{figure}
Faber \cite{8} also provided a characterization of strongly chordal graphs in terms of forbidden induced subgraphs. A \textit{trampoline} is a graph $G$ created by taking an even cycle $v_1, v_2, v_3, ..., v_{2k}, v_1$, and then adding edges between the vertices with even subscripts, so that the vertices $v_2, v_4, ...v_{2k}$ form a complete subgraph. The graph shown in Fig. 3 is an example of a trampoline with $k = 4$. Additionally, the graph depicted in  Fig. 8 is a trampoline with $k = 3$. A trampoline consisting of $2k$ vertices is also referred to as a $k$-\textit{sun}.\\
\begin{figure}[H]
    \centering
    \begin{tikzpicture}[line cap=round,line join=round,x=1.0cm,y=1.0cm,scale=1.]
\clip(-1.,0.) rectangle (5.,6.007295261072009);
\draw [line width=.2pt] (1.,4.)-- (3.,4.);
\draw [line width=.2pt] (1.,4.)-- (2.,5.);
\draw [line width=.2pt] (2.,5.)-- (3.,4.);
\draw [line width=.2pt] (1.,4.)-- (1.,2.);
\draw [line width=.2pt] (1.,2.)-- (3.,2.);
\draw [line width=.2pt] (3.,4.)-- (3.,2.);
\draw [line width=.2pt] (1.,4.)-- (3.,2.);
\draw [line width=.2pt] (1.,2.)-- (3.,4.);
\draw [line width=.2pt] (1.,2.)-- (2.,1.);
\draw [line width=.2pt] (3.,2.)-- (2.,1.);
\draw [line width=.2pt] (3.,4.)-- (4.,3.);
\draw [line width=.2pt] (3.,2.)-- (4.,3.);
\draw [line width=.2pt] (1.,4.)-- (0.,3.);
\draw [line width=.2pt] (0.,3.)-- (1.,2.);
\draw (1.8,5.49) node[anchor=north west] {$v_1$};
\draw (3.1,4.3) node[anchor=north west] {$v_2$};
\draw (4.,3.3) node[anchor=north west] {$v_3$};
\draw (3.,2.05) node[anchor=north west] {$v_4$};
\draw (1.8,1) node[anchor=north west] {$v_5$};
\draw (0.4,2.1) node[anchor=north west] {$v_6$};
\draw (-0.7,3.18) node[anchor=north west] {$v_7$};
\draw (0.3,4.4) node[anchor=north west] {$v_8$};
\begin{scriptsize}
\draw [fill=black] (1.,4.) circle (2.5pt);
\draw [fill=black] (3.,4.) circle (2.5pt);
\draw [fill=black] (2.,5.) circle (2.5pt);
\draw [fill=black] (1.,2.) circle (2.5pt);
\draw [fill=black] (3.,2.) circle (2.5pt);
\draw [fill=black] (2.,1.) circle (2.5pt);
\draw [fill=black] (4.,3.) circle (2.5pt);
\draw [fill=black] (0.,3.) circle (2.5pt);
\end{scriptsize}
\end{tikzpicture}
    \caption*{Fig. 3. The Sun graph $S_4$}\label{fig3}
    \label{fig:enter-label}
\end{figure}

\begin{theo}[\cite{8}]
A chordal graph $G$ is strongly chordal graph if and only if $G$ does not contain a trampoline as an induced subgraph .
\end{theo}

Monma, Reed and Trotter \cite{24} shown that every co-TT graph is strongly chordal i.e. also chordal graph. Thus a co-TT graph does not contain $C_n (n\geq 4)$ as an induced subgraph.  In this paper our main goal is to establish a forbidden induced subgraph characterization of co-TT graphs.

\section{Preliminaries}
\par Lekkerkerker and Boland \cite{22} provided an insightful characterization of interval graphs based on the concept of asteroidal triples (AT) of vertices. An independent set of three vertices in a graph forms an \textit{asteroidal triple} if it is possible to connect any pair of these vertices with a path that avoids the neighbors of the third vertex.\\
\begin{theo}[\cite{22}]
    A chordal graph $G$ is an interval graph if and only if it does not contain any asteroidal triple of vertices.
\end{theo}
Lekkerkerker and Boland also compiled a list of forbidden induced subgraphs for interval graphs. In fact, they established the following theorem.\\
 \begin{theo}[\cite{22}]
     A chordal graph is an interval graph if and only if it does not contain any graph (namely $T, W, T_n$ and $H_n$) of Fig.4 as an induced subgraph.
 \end{theo}
 
 \begin{figure}[H]
    \centering
    \begin{tikzpicture}[line cap=round,line join=round,x=1.0cm,y=1.0cm,scale=1.2]
\clip(0.6,.5) rectangle (12.,8.);
\draw [line width=.2pt] (1.,5.)-- (2.,5.);
\draw [line width=.2pt] (2.,5.)-- (3.,5.);
\draw [line width=.2pt] (3.,5.)-- (4.,5.);
\draw [line width=.2pt] (4.,5.)-- (5.,5.);
\draw [line width=.2pt] (3.,5.)-- (3.,6.);
\draw [line width=.2pt] (3.,6.)-- (3.,7.);
\draw [line width=.2pt] (7.,6.)-- (8.,6.);
\draw [line width=.2pt] (8.,6.)-- (9.,6.);
\draw [line width=.2pt] (9.,6.)-- (10.,6.);
\draw [line width=.2pt] (10.,6.)-- (11.,6.);
\draw [line width=.2pt] (7.,6.)-- (9.,5.);
\draw [line width=.2pt] (11.,6.)-- (9.,5.);
\draw [line width=.2pt] (10.,6.)-- (9.,5.);
\draw [line width=.2pt] (9.,6.)-- (9.,5.);
\draw [line width=.2pt] (8.,6.)-- (9.,5.);
\draw [line width=.2pt] (9.,7.)-- (9.,6.);
\draw [line width=.2pt] (3.,4.)-- (3.,3.);
\draw [line width=.2pt] (3.,3.)-- (2.,2.);
\draw [line width=.2pt] (3.,3.)-- (3.,2.);
\draw [line width=.2pt] (2.,2.)-- (1.,2.);
\draw [line width=.2pt] (2.482985430927987,1.9888405277485484)-- (3.,3.);
\draw [line width=.2pt] (2.,2.)-- (2.482985430927987,1.9888405277485484);
\draw [line width=.2pt] (2.482985430927987,1.9888405277485484)-- (3.,2.);
\draw [line width=.2pt] (3.,3.)-- (4.,2.);
\draw [line width=.2pt] (4.,2.)-- (5.,2.);
\draw [line width=.2pt] (9.,4.)-- (8.378698114060484,3.0157197149441486);
\draw [line width=.2pt] (9.,4.)-- (9.610088490987664,2.9864008964458826);
\draw [line width=.2pt] (8.378698114060484,3.0157197149441486)-- (9.610088490987664,2.9864008964458826);
\draw [line width=.2pt] (8.378698114060484,3.0157197149441486)-- (8.,2.);
\draw [line width=.2pt] (9.610088490987664,2.9864008964458826)-- (10.,2.);
\draw [line width=.2pt] (7.469814740614232,2.7958285762071515)-- (8.378698114060484,3.0157197149441486);
\draw [line width=.2pt] (7.469814740614232,2.7958285762071515)-- (8.,2.);
\draw [line width=.2pt] (9.610088490987664,2.9864008964458826)-- (10.418500562136417,2.793410786565484);
\draw [line width=.2pt] (10.418500562136417,2.793410786565484)-- (10.,2.);
\draw [line width=.2pt] (8.4,2.)-- (8.810375225420717,1.9893481182076334);
\draw [line width=.2pt] (8.,2.)-- (8.4,2.);
\draw [line width=.2pt] (8.,2.)-- (9.610088490987664,2.9864008964458826);
\draw [line width=.2pt] (8.378698114060484,3.0157197149441486)-- (8.4,2.);
\draw [line width=.2pt] (8.378698114060484,3.0157197149441486)-- (8.810375225420717,1.9893481182076334);
\draw [line width=.2pt] (8.810375225420717,1.9893481182076334)-- (9.610088490987664,2.9864008964458826);
\draw [line width=.2pt] (8.4,2.)-- (9.610088490987664,2.9864008964458826);
\draw [line width=.2pt] (8.378698114060484,3.0157197149441486)-- (10.,2.);
\draw (2.8,4.8) node[anchor=north west,scale=1.5] {$T$};
\draw (8.8,4.8) node[anchor=north west,scale=1.5] {$W$};
\draw (2.8,1.8) node[anchor=north west,scale=1.5] {$T_n $};
\draw (8.8,1.8) node[anchor=north west,scale=1.5] {$H_n$};
\begin{scriptsize}
\draw [fill=black] (1.,5.) circle (1.5pt);
\draw [fill=black] (2.,5.) circle (1.5pt);
\draw [fill=black] (3.,5.) circle (1.5pt);
\draw [fill=black] (4.,5.) circle (1.5pt);
\draw [fill=black] (5.,5.) circle (1.5pt);
\draw [fill=black] (3.,6.) circle (1.5pt);
\draw [fill=black] (3.,7.) circle (1.5pt);
\draw [fill=black] (7.,6.) circle (1.5pt);
\draw [fill=black] (8.,6.) circle (1.5pt);
\draw [fill=black] (9.,6.) circle (1.5pt);
\draw [fill=black] (10.,6.) circle (1.5pt);
\draw [fill=black] (11.,6.) circle (1.5pt);
\draw [fill=black] (9.,5.) circle (1.5pt);
\draw [fill=black] (9.,7.) circle (1.5pt);
\draw [fill=black] (3.,4.) circle (1.5pt);
\draw [fill=black] (3.,3.) circle (1.5pt);
\draw [fill=black] (2.,2.) circle (1.5pt);
\draw [fill=black] (3.,2.) circle (1.5pt);
\draw [fill=black] (1.,2.) circle (1.5pt);
\draw [fill=black] (2.482985430927987,1.9888405277485484) circle (1.5pt);
\draw [fill=black] (4.,2.) circle (1.5pt);
\draw [fill=black] (5.,2.) circle (1.5pt);
\draw [fill=black] (3.196910272007053,2) circle (.5pt);
\draw [fill=black] (3.3974765404761516,2) circle (.5pt);
\draw [fill=black] (3.597655378268306,2) circle (.5pt);
\draw [fill=black] (9.,4.) circle (1.5pt);
\draw [fill=black] (8.378698114060484,3.0157197149441486) circle (1.5pt);
\draw [fill=black] (9.610088490987664,2.9864008964458826) circle (1.5pt);
\draw [fill=black] (8.,2.) circle (1.5pt);
\draw [fill=black] (10.,2.) circle (1.5pt);
\draw [fill=black] (7.469814740614232,2.7958285762071515) circle (1.5pt);
\draw [fill=black] (10.418500562136417,2.793410786565484) circle (1.5pt);
\draw [fill=black] (8.4,2.) circle (1.5pt);
\draw [fill=black] (8.810375225420717,1.9893481182076334) circle (1.5pt);
\draw [fill=black] (9.2,2.) circle (.5pt);
\draw [fill=black] (9.4,2.) circle (.5pt);
\draw [fill=black] (9.6,2.) circle (.5pt);
\end{scriptsize}
\end{tikzpicture}
    \caption*{Fig. 4. These chordal graphs are not interval graphs because each of them contains an asteroidal triple (AT) of vertices. Additionally, each of the two lower graphs contains at least six vertices.}\label{fig4}
    \label{fig:enter-label}
\end{figure}

A bipartite graph $B=(X,Y,E)$ is a \textit{Ferrers bigraph} if the set of neighbors of any partite sets are linearly ordered by inclusion. An equivalent definition is, $B$ is a Ferrers bigraph if and only if $B$ does not contain $2K_2$ as an induced subgraph.\\
Let $A$ be a matrix. A stair partition of $A$ is a partition of its entries into two sets $(L,U)$ by a polygonal path from the upper left corner to the lower right corner of the matrix such that the set $L$ (respectively, $U$) is closed under leftward or downward (respectively, rightward or upward) movements. A $(0,1)$-matrix $A$ is a \textit{Ferrers matrix} if there exists a stair partition $(L,U)$ of $A$ such that all the entries in $L$ are 1 and those in $U$ are 0.\\

\begin{figure}[H]
    \centering
\begin{tikzpicture}[line cap=round,line join=round,,x=1.0cm,y=1.0cm, scale=1.5]
\clip(0.5,0) rectangle (8,2.5);
\draw (1,2)-- (2.5,2);
\draw (1,2)-- (1,0.5);
\draw (1,0.5)-- (2.5,0.5);
\draw (2.5,2)-- (2.5,0.5);
\draw (1,1.5)-- (1.5,1.5);
\draw (1.5,1.5)-- (1.5,1);
\draw (1.5,1)-- (2,1);
\draw (2,1)-- (2,0.5);
\draw (3.5,2)-- (5,2);
\draw (5,2)-- (5,0.5);
\draw (5,0.5)-- (3.5,0.5);
\draw (3.5,0.5)-- (3.5,2);
\draw (3.5,1.5)-- (4,1.5);
\draw (4,1.5)-- (4,1);
\draw (4,1)-- (4.5,1);
\draw (4.5,1)-- (4.5,0.5);
\draw (6,2)-- (7.5,2);
\draw (7.5,2)-- (7.5,0.5);
\draw (7.5,0.5)-- (6,0.5);
\draw (6,0.5)-- (6,2);
\draw (6,1.5)-- (6.5,1.5);
\draw (6.5,1.5)-- (6.5,1);
\draw (6.5,1)-- (7,1);
\draw (7,1)-- (7,0.5);
\draw (1.,1.25) node[anchor=north west,scale=1.5] {$L$};
\draw (1.75,1.75) node[anchor=north west,scale=1.5] {$U$};
\draw (3.5,1.25) node[anchor=north west,scale=1.5] {$1$};
\draw (4.25,1.75) node[anchor=north west,scale=1.5] {$0$};
\draw (6.,1.25) node[anchor=north west,scale=1.5] {$0$};
\draw (6.75,1.75) node[anchor=north west,scale=1.5] {$1$};
\draw (1.5,.5) node[anchor=north west,scale=1.] {$(i)$};
\draw (4,.5) node[anchor=north west,scale=1.] {$(ii)$};
\draw (6.5,.5) node[anchor=north west,scale=1.] {$(iii)$};

\end{tikzpicture}
    \caption*{Fig. 5. $(i)$ A stair partition $(L,U)$, $(ii)$ A Ferrers matrix $F$, $(iii)$ The complement of the Ferrers matrix $F$.}
    \label{fig:enter-label}
\end{figure}

Now $B$ is a Ferrers bigraph if and only if its biadjacency matrix $F$ is a Ferrers matrix (after suitable permutations of the rows and columns). Also it is easy to observe that the complement of $F$ is also a Ferrers matrix and $F$ does not contain  $\left( \begin{array}{ c  c } $1$ & $0$ \\ $0$ & $1$ \end{array} \right) or \left( \begin{array}{ c  c } $0$ & $1$ \\ $1$ & $0$ \end{array} \right)$ as a submatrix.

Several characterizations of interval bigraphs are known (see~\cite{17,27,30}). One characterization uses \textit{Ferrers bigraphs} ( introduced independently by Guttman \cite{15} and Riguet \cite{25} ), defined above.

Any binary matrix, specifically a (0, 1)-matrix, with only one zero is  a Ferrers matrix. Consequently, every bigraph $B$ can be represented as the intersection of a finite number of Ferrers bigraphs. The minimum number of Ferrers bigraphs required to intersect and reconstruct $B$ is referred to as the Ferrers dimension of the bigraph $B$, denoted as \textit{fdim}($B$).\par

Sen et al. \cite{27} proved that every interval bigraph has Ferrers dimension at most $2$, but the converse is not true. Bigraphs with Ferrers dimension $2$ were characterized independently by Cogis \cite{2} and others. Cogis \cite{2} introduced the \textit{associated graph} $H(B)$ for a bigraph $B$. Its vertices are the $0$'s of its biadjacency matrix of $B$, with two such vertices are adjacent in $H(B)$ if and only if they are the $0$'s of the $2$-by-$2$ permutation matrix. Cogis \cite{2} proved that $fdim(B) \leq 2$ if and only if $H(B)$ is bipartite. Sen et al. \cite{27} translate Cogis's condition to an adjacency matrix condition for a bigraph $B$ to be of $fdim(B) \leq 2$ in the following theorem.\par

\begin{theo}[\cite{27}]
A bigraph B has Ferrers dimension at most $2$ if and only if the rows and the columns of the biadjacency matrix $A(B)$ of $B$ can be permuted independently, so that no 0 has a 1 both below it and to its right.
\end{theo}
An equivalent definition of bigraphs of Ferrers dimension 2 has given by Hell et al. \cite{19}, which they called \textit{min-orderable bigraphs}. A bigraph is of \textit{min-orderable} if its biadjacency matrix has independent permutation of rows and columns such that the biadjacency matrix does not have any of the  matrices  $\left( \begin{array}{ c  c } $1$ & $0$ \\ $0$ & $1$ \end{array} \right) or \left( \begin{array}{ c  c } $0$ & $1$ \\ $1$ & $1$ \end{array} \right)$ as a submatrix.

A bigraph $B = (X, Y, E)$ is classified as an interval containment bigraph if, for each vertex $v \in X \cup Y$ of $B$, there exists an interval $I_v$ such that the vertices $x \in X$ and $y \in Y$ are adjacent if and only if the interval $I_x$ contains $I_y$. It has been noted in \cite{6, 21} that interval containment bigraphs are equivalent to bigraphs with Ferrers dimension 2.\\
\textit{Circular arc graphs} are the intersection graphs of a collection of circular arcs of a host circle. If the vertices of the circular arc graph $G$ can be partitioned into two disjoint cliques, denoted as $K$ and $K'$, it is referred to as a \textit{two-clique circular arc graph}. Let the vertex sets of $K$ and $K'$ be respectively denoted as $X$ and $Y$. The complement of $G$, denoted as $\Bar{G}$, forms a bipartite graph $B = (X, Y, E)$, where $E$ represents the set of non-edges in $G$.\\
\par Now, we present a result from J. Huang \cite{21} that establishes a connection between two-clique circular arc graphs and bigraphs with Ferrers dimension 2.
\begin{theo}[\cite{21}]\label{t7}
    A graph $G$ is a two-clique circular arc graph if and only if its complement $\Bar{G}$ is a bigraph of Ferrers dimension 2.
\end{theo}
As mentioned in \cite{19}, several classes of bigraphs, such as interval containment bigraphs, the complements of two-clique circular arc graphs, and two-directional orthogonal ray graphs introduced in \cite{28}, have been found to be equivalent to the class of signed-interval bigraphs. Thus signed-interval bigraphs are equivalent to  the class of bigraphs equivalent to Ferrers dimension 2.

In the following section, we provide a direct proof of the equivalence between the class of signed-interval bigraphs and the class of bigraphs with Ferrers dimension 2. In this context, we also characterize the adjacency matrix of signed-interval graphs.
\section{Biadjacency (adjacency) matrix characterization of signed-interval bigraphs (graphs)}

 The definition of signed-interval graphs has also been extended  to the class of bipartite graphs. A bigraph $B$=$(X, Y, E)$ is a \textit{signed interval bigraph}, if corresponding to each vertex $v \in X \cup Y$ we can assign a positive interval or negative interval such that $xy \in E$ if and only if ${I_x}^+ \cap {I_y}^+ \neq \emptyset$ or ${I_x}^- \subseteq {I_y}^+$ or ${I_y}^- \subseteq {I_x}^+$ .\par
As stated before, in the following theorem, we proved that signed-interval bigraphs are equivalent to the class of bigraphs of Ferrers dimension 2. Our proof is constructive, meaning that when a bigraph has Ferrers dimension 2, we can create its representation as a signed-interval bigraph.
\begin{theo}
For a bigraph $B= (X, Y, E)$ the following conditions are equivalent.
\begin{enumerate}[(a)]
\item $B$ is a signed-interval bigraph.
\item The rows and columns of the biadjacency matrix $A(B)$ of $B$ can be permuted independently so that no $0$ has a $1$ both below and to its right.
\item $B$ is a bigraph of Ferrers dimension at most $2$.
\end{enumerate}
\end{theo}

\noindent \textbf{Proof.}  $(a) \Rightarrow (b)$ Let $B$ be a signed-interval bigraph. Without loss of generality we may assume that all the left end point of the intervals are distinct in the representation of $B$. Now, we arrange the rows and columns of $A(B)$ respectively according to the increasing order of the left end points of $I_{x_k}$, $x_k \in X$ and $I_{y_l}$, $y_l \in Y$. Next, suppose that $(i, j)$th entry of $A(B)$ is zero. We consider the following cases.
\begin{case}
Suppose ${I_{x_i}}^+ \cap {I_{y_j}}^+ = \emptyset$. Then we have $r(x_i) < l(y_j)$ or $r(y_j) < l(x_i)$. The first possibility implies for all ${I_{{y_k}}}^+$ $(k>j)$, ${I_{x_i}}^+ \cap {I_{y_k}}^+ = \emptyset$ and for all ${I_{{y_k}}}^-$ $(k>j)$, $ {I_{y_k}}^- \nsubseteq {I_{x_i}}^+ $. Thus $x_iy_k = 0$ for all $k \geqslant j$. In other words all the entries to the right of $(i, j)$th entry are $0$s. Again, if $r(y_j) < l(x_i)$, then for all ${I_{{x_k}}}^+$ $(k>i)$, ${I_{x_i}}^+ \cap {I_{y_j}}^+ = \emptyset$. Also for all ${I_{{x_k}}}^-$ $(k>i)$, $ {I_{x_k}}^- \nsubseteq {I_{y_j}}^+ $. So all the entries are zeros below $(i, j)$th entry.
\end{case}
\begin{case}
Suppose ${{I_y}_j}^- \nsubseteq {{I_x}_i}^+$. Now we have the following possibilities: either $r(y_j) < l(x_i)$ or $l(y_j) > r(x_i)$. In the first possibility, for all $x_k$ $(k > i)$, $ {I_{y_j}}^- \nsubseteq {I_{x_k}}^+ $. So all the entries are zero below $(i, j)$th entry. In the other possibility, for all $y_k$ $(k>j)$, ${I_{y_k}}^+ \cap {I_{x_i}}^+ = \emptyset$ or $ {I_{y_k}}^- \nsubseteq {I_{x_i}}^+ $. Thus all the entries to the right of $(i, j)$th entry are $0$s.
\end{case}
\begin{case}
Suppose ${{I_x}_i}^- \nsubseteq {{I_y}_j}^+$. Then we have the following possibilities: either $r(x_i) < l(y_j)$ or $r(y_j) <l(x_i)$. As before, in the first possibility all the entries to the right of $(i, j)$th entry are $0$s and in the second possibility all the entries below $(i, j)$th entry are $0$s.
\end{case}
Thus, with this arrangement of rows and columns of $A(B)$, no $0$ of $A(B)$ has a $1$ both below and to its right.\par
$(b) \Leftrightarrow (c)$ This follows directly from Theorem 5.

$(c) \Rightarrow (a)$ Let $B = (X, Y, E)$ be a bigraph of Ferrers dimension $2$. Also suppose that the vertices can be arranged as $x_1$, $x_2$, \ldots , $x_n$ and $y_1$, $y_2$, \ldots , $y_m$ respectively in the rows and columns of the biadjacency matrix $A(B)$ of $B$ such that no $0$ has a $1$ both below and to its right, where $x_i \in X$ $(1 \leq i \leq n)$ and $y_j \in Y$ $(1 \leq j \leq m)$.\par

Next, we assign intervals to the vertices of $X$ as follows. Let $x_i$ be any vertex of $X$ and in the $x_i$th row suppose the last $1$ appears in the $k$ th column. Then the interval corresponding to $x_i$ is $[i, k]$ and we denote it by ${{I_x}_i}^+$ if $i \leq k$ or ${{I_x}_i}^-$ if $i > k$. In the similar way we assign an interval $[j, l]$ corresponding to the vertex $y_j$ . It will be denoted by ${{I_y}_j}^+$ if $j \leq l$ or ${{I_y}_j}^-$ if $j > l$.\par

Now, we shall prove that with this interval assignment, $B$ is a signed interval bigraph. Let $x_iy_j \in E$, i.e. $(i, j)$th entry of $A(B)$ is $1$. Let the interval corresponding to $x_i$ is $[i, k]$ and  the interval corresponding to $y_j$ is $[j, l]$. Again $(i, j)$th entry of is $1$ implies $j \leq k$ and $i \leq l$. Now consider the three following cases.\par \vspace{.2cm}
\noindent\textbf{Case~1.\ }Suppose ${{I_x}_i}$ and ${{I_y}_j}$ are both positive intervals. Then  $j \leq k$ and $i \leq l$ imply ${I_{x_i}}^+ \cap {I_{y_j}}^+ \neq \emptyset$.\par \vspace{.2cm}
\noindent\textbf{Case~2.\ }Suppose ${{I_x}_i}$ be a positive interval and ${{I_y}_j}$ be a negative interval. Then  $i \leq k$ and $j > l$. Also $j \leq k$ and $i \leq l$ imply ${I_{y_j}}^-  \subseteq {I_{x_i}}^+$.\par \vspace{.2cm}
\noindent\textbf{Case~3.\ }Suppose ${{I_x}_i}$ be a negative interval and ${{I_y}_j}$ be a positive interval. Then  $j \leq l$ and $i>k$. Again $j \leq k$ and $i \leq l$ imply ${I_{x_i}}^-  \subseteq {I_{y_j}}^+$.\par \vspace{.2cm}

Next, suppose that $(i, j)$th entry is $0$ and all the entries to its right are $0$s. Then $k < j$ and let $l > i$. Consider the following possibilities. If ${{I_x}_i}$ and ${{I_y}_j}$ are both positive intervals, then  $k < j$ implies ${I_{x_i}}^+ \cap {I_{y_j}}^+ = \emptyset$. Again, if ${{I_x}_i}$ is a negative interval and ${{I_y}_j}$ is a positive interval, then $k < j$ implies ${I_{x_i}}^-  \nsubseteq {I_{y_j}}^+$. Next if ${{I_x}_i}$ is positive interval and ${{I_y}_j}$ is a negative interval, then again $k < j$ implies ${I_{y_j}}^-  \nsubseteq {I_{x_i}}^+$.\par

Finally, suppose that $(i, j)$th entry is $0$ and all the entries are $0$s below it. Then we can assume that $k > j$ and $l < i$. Then as before we can show that ${I_{x_i}}^+ \cap {I_{y_j}}^+ = \emptyset$. Also if one of the two intervals, say ${{I_x}_i}$ is a negative interval then ${I_{x_i}}^-  \nsubseteq {I_{y_j}}^+$. This completes the proof of the theorem.\QEDA

Signed-interval bigraphs are, therefore, a generalization of interval bigraphs. Interval bigraphs specifically refer to those signed-interval bigraphs where all vertices are assigned positive intervals.\par

In the next theorem,  we provide an analogous characterization of signed-interval graphs, introduced earlier, in terms of their adjacency matrices. It is important to note that in our assumption for signed-interval graphs, self-loops are present for vertices where positive intervals are assigned. Therefore, among the diagonal entries of the adjacency matrix of a signed-interval graph, we find a '1' corresponding to vertices with positive intervals are assigned, while the rest are '0s'. \\
We define a graph $G$ is of \textit{Ferrers dimension 2} if and only if there is an arrangement of the rows of the adjacency matrix $A(G)$ (and same arrangement for the columns) such that no $0$ has a $1$ both below and to its right. In the next theorem we show that signed-interval graphs are equivalent to the graphs of Ferrers dimension 2. Also as in Theorem 7, our proof is constructive. 
\par

\begin{theo}
For a graph $ G= (V, E)$ the following conditions are equivalent.
\begin{enumerate}[(a)]
\item G is a signed-interval graph (i.e. co-TT graph).
\item There is a permutation of the rows (and same permutation of the columns) such that the adjacency matrix $A(G)$ of $G$ is such that no $0$ has a $1$ both below and to its right, where all the diagonal entries of $A(G)$ are not all $1$ or $0$s. 
\item $G$ is a graph of Ferrers dimension $2$.
\end{enumerate}
\end{theo}
\noindent\textbf{Proof.} $(b) \Leftrightarrow (c)$ is trivial.\par
$(a) \Rightarrow (b)$ Let G be a signed-interval graph. Without loss of generality we may assume that the left end point of the intervals assigned to the vertex set $V$ are distinct. We arrange the rows(columns) of the adjacency matrix $A = A(G)$ of $G$ according to the increasing order of the left end points of the corresponding intervals. Next, we shall show that the matrix $A$ is such that if the $(i, j)$th entry of $A$ is a zero then each entry to the right of it is also a zero or each entry below it is also zero.\par
Now consider the following cases.
\begin{case}
Suppose ${I_{v_i}}^+ \cap {I_{v_j}}^+ = \emptyset$ and $i < j$, then $r(v_i) < l(v_j)$. Now $l(v_k) > r(v_i)$, $\forall\ k \geq j$. This implies that all the entries to the right of $(i, j)$th entry are also $0$s. Also all the entries to the below of $(j, i)$th entry are $0$s. Next, assume ${I_{v_i}}^+ \cap {I_{v_j}}^+ =\emptyset$ and $i>j$, then $r(v_j)<l(v_i)$. Similarly we conclude that all the entries below $(i,j)$th entry are $0$s and all the entries to the right of $(j,i)$th entry are $0$s.
\end{case}
\begin{case}
Let ${I_{v_j}}^- \nsubseteq {I_{v_i}}^+ $ and $i < j$, then we have either $l(v_j) > r(v_i)$ or $r(v_j) < l(v_i)$. Then in the first possibility we have $l(v_k) > r(v_i)$, $\forall k \geq j$. So all the entries to the right of $(i, j)$th entry are $0$s and all the entries below $(j,i)$th entry are $0$s. In the other possibility for all $k \geq i$, ${I_{v_j}}^- \nsubseteq {I_{v_k}}^+ $. Thus all the entries below $(i, j)$th entry are $0$s and all the entries to the right of $(j,i)$the entry are $0$s.
\end{case}
\begin{case}
Let ${I_{v_i}}^- \nsubseteq {I_{v_j}}^+ $ and $j < i$, then we have either $r(v_i) < l(v_j)$ or $r(v_j) < l(v_i)$. In the first possibility $\forall k > j$, ${I_{v_i}}^- \nsubseteq {I_{v_k}}^+ $ and hence all the entries to the right of $(i, j)$th entry are $0$s. Also $(j, i)$th entry is a zero and all the entries below $(j, i)$th entry are $0$s. Similarly in the second possibility we can show that all the entries to the right of $(j, i)$th entry are $0$s and all the entries below $(i, j)$th entry are $0$s.
\end{case}
It can be easily observed that the matrix $A$ is symmetric. Without loss of generality may assume the $(i, j)$th entry is $1$. Then we have either ${I_{v_i}}^+ \cap {I_{v_j}}^+ \neq \emptyset$ or ${I_{v_i}}^- \subseteq {I_{v_j}}^+$ or ${I_{v_j}}^- \subseteq {I_{v_i}}^+ $. In any case $(j, i)$th entry is also a  $1$ and if the $(i, j)$th entry of $A$ is $0$, then we have  ${I_{v_i}}^+ \cap {I_{v_j}}^+ = \emptyset$ or ${I_{v_j}}^- \nsubseteq {I_{v_i}}^+ $ or  ${I_{v_i}}^- \nsubseteq {I_{v_j}}^+ $. Then also $(j, i)$th entry of $A$ is also $0$.\par
Now by Theorem~$2$, $A = A(G)$ is a matrix of Ferrers dimension $2$, where all the diagonal entries are not $1$ or $0$s as the intervals are not all positive or negative. \par
$(b) \Rightarrow (a)$ Let the adjacency matrix $A=A(G)$ of $G$ is of Ferrers dimension $2$, where all diagonal entries are not $1$ or $0$s. Thus we can arrange the rows(columns) of $A$ such that no $0$ has a $1$ both to its right and below it. Now, we can construct a signed-interval representation of $G$ as described in the Theorem $8$.\QEDA \vspace{.5cm}
\par It may be noted that when all the diagonal entries are one then a signed-interval graph reduces to an interval graph.

\section{Forbidden induced subgraphs of co-TT graphs}
In this section, we will provide a characterization of co-TT graphs based on forbidden induced subgraphs. This problem was originally posed by Monma, Reed, and Trotter \cite{24} and has also been mentioned in \cite{11, 14}.\\
In \cite{14}, Golumbic et al. characterized split co-TT graphs in terms of forbidden induced subgraphs. Additionally, Golovach et al. \cite{11} introduced a recognition algorithm for co-TT graphs with a time complexity of O($n^2$).
\par As previously mentioned, if $G = (V, E)$ is a two-clique circular arc graph, then its complement $\Bar{G}$, can be represented as the bipartite graph $B = (X, Y, E)$, where E consists of the non-edges of $G$.
\par Trotter and Moore \cite{29} characterized two-clique circular arc graphs by identifying specific forbidden induced subgraphs. They presented these forbidden subgraphs as a set system and established that a graph $G$ is a two-clique circular arc graph if and only if its complement, $\Bar{G}$, does not contain any induced subgraphs of the form $G_1, G_2, G_3$, and several infinite families $C_i (i \geq 3), T_i, W_i, D_i, M_i, N_i (i \geq 1)$, as illustrated in  Fig. 6.\\
\begin{figure}[H]\label{fig7}
$\mathcal{C}_3 = \{\{1,2\},\{2,3\},\{3,1\}\}$ \\
$\mathcal{C}_4 = \{\{1,2\},\{2,3\},\{3,4\},\{4,1\}\}$ \\
$\mathcal{C}_5 = \{\{1,2\},\{2,3\},\{3,4\},\{4,5\},\{5,1\}\}$ \\
\dots \\
$\mathcal{T}_1 = \{\{1,2\},\{2,3\},\{3,4\},\{2,3,5\},\{5\}\}$ \\
$\mathcal{T}_2 = \{\{1,2\},\{2,3\},\{3,4\},\{4,5\},\{2,3,4,6\},\{6\}\}$ \\
$\mathcal{T}_3 = \{\{1,2\},\{2,3\},\{3,4\},\{4,5\},\{5,6\},\{2,3,4,5,7\},\{7\}\}$ \\
\dots \\
$\mathcal{W}_1 = \{\{1,2\},\{2,3\},\{1,2,4\},\{2,3,4\},\{4\}\}$ \\
$\mathcal{W}_2 = \{\{1,2\},\{2,3\},\{3,4\},\{1,2,3,5\},\{2,3,4,5\},\{5\}\}$ \\
$\mathcal{W}_3 = \{\{1,2\},\{2,3\},\{3,4\},\{4,5\},\{1,2,3,4,6\},\{2,3,4,5,6\},\{6\}\}$ \\
\dots \\
$\mathcal{D}_1 = \{\{1,2,5\},\{2,3,5\},\{3\},\{4,5\},\{2,3,4,5\}\}$ \\
$\mathcal{D}_2 = \{\{1,2,6\},\{2,3,6\},\{3,4,6\},\{4\},\{5,6\},\{2,3,4,5,6\}\}$ \\
$\mathcal{D}_3 =\{\{1,2,7\},\{2,3,7\},\{3,4,7\},\{4,5,7\},\{5\},\{6,7\},\{2,3,4,5,6,7\}\}$ \\
\dots \\
$\mathcal{M}_1 = \{\{1,2,3,4,5\},\{1,2,3\},\{1\},\{1,2,4,6\},\{2,4\},\{2,5\}\}$ \\
$\mathcal{M}_2 = \{\{1,2,3,4,5,6,7\},\{1,2,3,4,5\},\{1,2,3\},\{1\},\{1,2,3,4,6,8\},\{1,2,4,6\},\{2,4\},\{2,7\}\}$ \\
$\mathcal{M}_3 =\{\{1,2,3,4,5,6,7,8,9\},\{1,2,3,4,5,6,7\},\{1,2,3,4,5\},\{1,2,3\},\{1\}, \{1,2,3,4,5,6,8,10\},$

\hfill
$\{1,2,3,4,6,8\},\{1,2,4,6\},\{2,4\},\{2,9\}\}$ \\
\dots \\
$\mathcal{N}_1 = \{\{1,2,3\},\{1\},\{1,2,4,6\},\{2,4\},\{2,5\},\{6\}\}$ \\
$\mathcal{N}_2 = \{\{1,2,3,4,5\},\{1,2,3\},\{1\},\{1,2,3,4,6,8\},\{1,2,4,6\},\{2,4\},\{2,7\},\{8\}\}$ \\
$\mathcal{N}_3 = \{\{1,2,3,4,5,6,7\},\{1,2,3,4,5\},\{1,2,3\},\{1\},\{1,2,3,4,5,6,8,10\},\{1,2,3,4,6,8\},$

\hfill 
$\{1,2,4,6\},\{2,4\},\{2,9\},\{10\}\}$ \\
\dots \\
$\mathrm{G}_1 = \{\{1,3,5\},\{1,2\},\{3,4\},\{5,6\}\}$ \\
$\mathrm{G}_2 = \{\{1\},\{1,2,3,4\},\{2,4,5\},\{2,3,6\}\}$ \\
$\mathrm{G}_3 = \{\{1,2\},\{3,4\},\{5\},\{1,2,3\},\{1,3,5\}\}$\\
 \caption*{Fig. 6. Forbidden families of two-clique circular arc graphs as in \cite{29}.}
    \label{}
\end{figure}

\par Feder, Hell, and Huang \cite{9} provided an explanation of derivation of the forbidden bipartite graphs from Fig. 6. These bipartite graphs are presented in Fig. 7.\\
\begin{figure}[H]
\centering
\begin{minipage}{.5\textwidth}
  \centering
    \begin{tikzpicture}[scale=.7]
\tikzstyle{every node}=[circle, draw, fill=black,
                        inner sep=1pt, minimum width=2.5pt]
 \node (a1) at (1,3){};
    \node(a2) at(1,4){};
    \node(a3) at (2,3){};
    \node(a4) at (2,4){};
    \node(a5) at (1,2){};
    \node(a6) at (2,2){};
    \node(a7) at(1,1){};
    \node(a8) at (2,1){};
    \node(a9) at (0,2){};
    \node(a10) at (0,3){};
    \node(a11) at(3,3){};
    \node(a12) at (3,2){};
    \draw (a1)--(a2);
    \draw (a3)--(a4);\draw(a2)--(a4);
    \draw(a1)--(a3);
    \draw(a3)--(a11);
    \draw(a11)--(a12);
    \draw(a12)--(a6);
    \draw(a6)--(a5);
    \draw(a5)--(a9);
    \draw(a9)--(a10);
    \draw(a10)--(a1);
    \draw(a1)--(a12);
    \draw(a3)--(a9);
    \draw(a1)--(a5);
    \draw(a3)--(a6);
    \draw(a6)--(a8);
    \draw(a5)--(a7);
    
    \end{tikzpicture}\\
   $\mathcal{M}_1 $
\end{minipage}%
\begin{minipage}{.5\textwidth}
  \centering
 \begin{tikzpicture}[scale=.7]
\tikzstyle{every node}=[circle, draw, fill=black,
                        inner sep=1pt, minimum width=2.5pt]
  \node(a) at(7,7){};
   \node(b) at (8,7){};
   \node(c) at (7,8){};
   \node(d) at (8,8){};
    \node(e) at (5,6){};
   \node(f) at (5,5){};
   \node(g) at (5,4){};
   \node(h) at (5,3){};
   \node(i) at (7,2){};
   \node(j) at (8,2){};
   \node(k) at (10,3){};
   \node(l) at (10,4){};
   \node(m) at (10,5){};
   \node(n) at (10,6){};
   \node(o) at (7,1){};
   \node(p) at (8,1){};
   \draw(f)--(m);
   \draw(a)--(b);
   \draw(a)--(c);
   \draw(c)--(d);
   \draw(d)--(b);
   \draw(a)--(e);
   \draw(e)--(f);
   \draw(f)--(g);
   \draw(g)--(h);
   \draw(h)--(i);
   \draw(i)--(j);
   \draw(j)--(k);
   \draw(k)--(l);
   \draw(l)--(m);
   \draw(m)--(n);
   \draw(n)--(b);
   \draw(i)--(o);
   \draw(j)--(p);
   \draw(a)--(g);
   \draw(a)--(i);
   \draw(a)--(k);
   \draw(a)--(m);
   \draw(b)--(f);
   \draw(b)--(h);
   \draw(b)--(j);
   \draw(b)--(l);
   \draw(f)--(i);
   \draw(g)--(j);
   \draw(i)--(l);
   \draw(j)--(m);
  \draw[dashed](11,4.5)--(13,4.5);

    \end{tikzpicture}\\
   $\mathcal{M}_2 $
\end{minipage}

\end{figure}
\begin{figure}[H]
\centering
\begin{minipage}{.5\textwidth}
  \centering
    \begin{tikzpicture}[scale=.7]
\tikzstyle{every node}=[circle, draw, fill=black,
                        inner sep=1pt, minimum width=2.5pt]
\node(a) at (0,1){};
\node(b) at (1,1){};
\node(c) at (2,1){};
\node(d) at (2,0){};
\node(e) at (3,0){};
\node(f) at (3,1){};
\node(g) at (4,1){};
\node(h) at (5,1){};
\node(i) at (2,2){};
\node(j) at (2,3){};
\node(k) at (3,2){};
\node(l) at (3,3){};
\draw(a)--(b);
\draw(b)--(c);
\draw(c)--(d);
\draw(d)--(e);
\draw(e)--(f);
\draw(f)--(g);
\draw(g)--(h);
\draw(c)--(f);
\draw(c)--(i);
\draw(i)--(j);
\draw(i)--(k);
\draw(k)--(l);
\draw(f)--(k);
 
    \end{tikzpicture}\\
   $\mathcal{N}_1 $
\end{minipage}%
\begin{minipage}{.5\textwidth}
  \centering
 \begin{tikzpicture}[scale=.7]
\tikzstyle{every node}=[circle, draw, fill=black,
                        inner sep=1pt, minimum width=2.5pt]
\node(a) at (5,6){};
\node(b) at (6,6){}; 
\node(c) at (4,5){};
\node(d) at (4,4){};
\node(e) at (4,3){};
\node(f) at (5,2){};
\node(g) at (6,2){};
\node(h) at (7,3){};
\node(i) at (7,4){};
\node(j) at (7,5){};
\node(k) at (8,6){};
\node(l) at (3,6){};
\node(m) at (8,4){};
\node(n) at (9,4){};
\node(o) at (3,4){};
\node(p) at (2,4){};
\draw(a)--(b);
\draw(a)--(c);
\draw(c)--(d);
\draw(c)--(l);
\draw(d)--(o);
\draw(o)--(p);
\draw(d)--(e);
\draw(e)--(f);
\draw(f)--(g);
\draw(g)--(h);
\draw(h)--(i);
\draw(i)--(m);
\draw(m)--(n);
\draw(i)--(j);
\draw(j)--(k);
\draw(j)--(b);
\draw(a)--(i);
\draw(b)--(d);
\draw(c)--(j);
\draw(c)--(f);
\draw(d)--(i);
\draw(d)--(g);
\draw(f)--(i);
\draw(g)--(j);
\draw[dashed](10,4)--(12,4);
    \end{tikzpicture}\\
   $\mathcal{N}_2 $
\end{minipage}

\end{figure}
\begin{figure}[H]
    \centering
    \begin{tikzpicture}[scale=1]
\tikzstyle{every node}=[circle, draw, fill=black,
                        inner sep=0pt, minimum width=2.5pt]
\node(a) at (4,5){};
\node(b) at (5,5){};
\node(c) at (6,4){};
\node(d) at (6,3){};
\node(e) at (5,2){};
\node(f) at (4,2){};
\node(g) at (3,3){};
\node(h) at (3,4){};
\draw (g)--(h);
\draw(h)--(a);
\draw(a)--(b);
\draw(b)--(c);
\draw(c)--(d);
\draw(d)--(e);
\draw(e)--(f);
\draw[dashed](f)--(g);
   \end{tikzpicture}\\
    $\mathcal{C}_i $
\end{figure}

\begin{figure}[H]
\centering
\begin{tikzpicture}[scale=1]
{\tikzstyle{every node}=[circle, draw, fill=black,
                        inner sep=0pt, minimum width=2.5pt]
\draw (0:0) node(a){}--++(100:1)node(b){}--++(150:1)node(c){}--++(100:-1)node(d){}--cycle;
\draw (a)--++(0:.6)node{}--++(0:.6)node(e){};
\draw[dashed] (e)--++(0:1)node(f){};
\draw (f)--++(80:1)node(g){}--++(30:1)node(h){}--++(80:-1)node(i){}--++(30:-1) (b)--(e) (b)--(f) (g)--(a) (g)--(e);
\path [name path=l] (b) --++ (30:2);
\path [name path=m] (g) --++ (-30:-2);
\path [name intersections={of=l and m, by=j}];
\draw (b)--(j)node{}--(g) (j)--++(90:.7)node(k){};
\draw[line width=.4mm] (c)--(d) (h)--(i) (j)--(k);

\coordinate (a1) at ([xshift=2.4cm]f);
\draw (a1)node{}--++(0:.6)node(b1){}--++(90:1)node(c1){}--++(0:1.6)node(d1){}--++(90:1)node(e1){}--++(0:-1.6)node(f1){}--++(90:-1) (b1)--++(0:.6)node(g1){}--(d1)--(a1);
\draw[dashed] (g1)--++(0:1.6)node(h1){};
\draw (h1)--(c1);
\path[name path=i1] (a1)--++(100:1);
\path[name path=j1] (c1)--++(30:-1);
\path[name intersections={of=i1 and j1, by=k1}];
\path[name path=l1] (h1)--++(80:1);
\path[name path=m1] (d1)--++(-30:1);
\path[name intersections={of=l1 and m1, by=n1}];
\draw (a1)--(k1)node{}--(c1) (k1)--++(110:1)node(o1){};
\draw (h1)--(n1)node{}--(d1) (n1)--++(70:1)node(p1){};
\draw[line width=.4mm] (k1)--(o1) (n1)--(p1) (e1)--(f1);
\coordinate (z1) at ([xshift=.8cm]b1);

\coordinate (a2) at ([xshift=1.2cm]h1);
\draw (a2)node{}--++(0:.6)node(b2){}--++(0:.6)node(c2){}--++(0:.6)node{}--++(0:.6)node(d2){};
\draw[dashed] (d2)--++(0:1)node(e2){};
\draw (e2)--++(0:.6)node(f2){}--++(0:.6)node(g2){};
\path[name path=h2] (c2)--++(45:2);
\path[name path=i2] (e2)--++(-45:-2);
\path[name intersections={of=h2 and i2, by=j2}];
\draw (c2)--(j2)node{}--(e2) (j2)--(d2) (j2)--++(90:.6)node(k2){}--++(90:.6)node(l2){};
\draw[line width=.4mm] (a2)--(b2) (f2)--(g2) (k2)--(l2);

\coordinate (a3) at ([yshift=-3.6cm]d);
\draw (a3)node{}--++(0:.6)node(b3){}--++(0:.6)node{}--++(0:.6)node(c3){}--++(90:.7)node{}--++(90:.7)node(d3){}--++(90:.7)node(e3){} (c3)--++(0:.6)node{}--++(0:.6)node(f3){}--++(0:.6)node(g3){};
\draw[line width=.4mm] (a3)--(b3) (e3)--(d3) (f3)--(g3);

\coordinate (a4) at ([xshift=1.5cm]g3);
\draw (a4)node{}--++(0:.6)node(b4){}--++(0:.6)node(c4){}--++(0:1.2)node(d4){}--++(0:.6)node(e4){}--++(0:.6)node(f4){} (c4)--++(90:1)node{}--++(90:1)node(g4){}--++(0:1.2)node(h4){}--++(90:-1)node{}--++(0:-1.2) (d4)--++(90:1);
\draw[line width=.4mm] (a4)--(b4) (e4)--(f4) (g4)--(h4);
\coordinate (z2) at ([xshift=.6cm]c4);

\coordinate (a5) at ([xshift=1.5cm]f4);
\draw (a5)node{}--++(0:.7)node(b5){}--++(90:.8)node{}--++(0:.9)node(c5){}--++(0:.9)node{}--++(90:-.8)node(d5){}--++(0:.7)node(e5){} (b5)--++(0:.9)node(f5){}--(e5) (f5)--(c5)--++(90:.7)node(g5){}--++(90:.7)node(h5){};
\draw[line width=.4mm] (a5)--(b5) (d5)--(e5) (g5)--(h5);
}

\coordinate (z) at (0,-0.5);
\draw (z-|j)node{$\mathcal{W}_i$} (z-|z1)node{$\mathcal{D}_i$} (z-|j2)node{$\mathcal{T}_i$};
\coordinate (y) at ([yshift=-.5cm]c3);
\draw (y)node{$\mathrm{G}_1$} (y-|z2)node{$\mathrm{G}_3$} (y-|c5)node{$\mathrm{G}_2$};
\end{tikzpicture}\\ \vspace{.5cm}
Fig. 7. Forbidden bigraphs from Fig. 6.\\
\end{figure}\vspace{-.4cm}

In the same paper \cite{9} they also introduced the notion of edge-asteroid to unify the families of graphs depicted in Fig. 7. An \textit{edge-asteroid} is a set of $2k+1$ edges $e_0, e_1, ..., e_{2k}$ such that, for $i= 0, 1, ..., 2k$ there is a path joining $e_i$ and $e_{i+1}$ and containing both $e_i$ and $e_{i+1}$ that avoids the neighbors of $e_{i+k+1}$; the subscript addition is modulo $2k+1$. Then they characterized two-clique circular arc graphs in the following theorem.
\begin{theo}[\cite{9}]\label{t10}    A graph is a two clique circular arc graph if and only if its complement does not contain induced cycle of length at least six and an edge-asteroid.
\end{theo}
It can be easily checked that each graph of the families $T_i, W_i, D_i, M_i, N_i (i\geq 1)$  as well as the graphs  $G_1, G_2, G_3$ contains an edge asteroid.
Consequently, from \cref{t7} and \cref{t10}, it can be deduced that the bipartite graphs depicted in Fig. 7 serve as the forbidden induced subgraphs for bipartite graphs with Ferrers dimension 2. Therefore, they are also the forbidden induced subgraphs for signed-interval bigraphs.\\

\par A graph $G$ is a \textit{split graph} if its vertex set $V(G)$ can be partitioned into two disjoint subsets, where one is a stable set (a set of vertices where no two vertices are adjacent) and other is a complete graph. \\
\par The following theorem provides two characterizations of split graphs.
\begin{theo}[\cite{12}]
    A graph $G$ is a split graph if and only if $G$ and its complement $\Bar{G}$ are chordal graphs or equivalently if and only if $G$ does not contain $C_n (n\geq 4)$ as an induced cycle.
   
\end{theo}
The second characterization can be derived from the fact that $G$ is a split graph if and only if $G$ does not contain $C_4, C_5$, and $2K_2$ as induced subgraphs.\\
Golumbic, Weingarten, and Limouzy \cite{14} introduced the class of split co-TT graphs (i.e. the intersection of the classes of split graphs and co-TT graphs). In other words the graphs which belong to both the classes of split graphs and also co-TT graphs. Also they have provided the forbidden induced subgraphs of split co-TT graphs by building upon the work of Trotter and Moore \cite{29}, as follows.\\
Let $B=(X,Y,E)$ be any bigraph of Fig. 7. Then its adjacency matrix $A(B)$ can be written in the form 
\begin{figure}[H]
    \centering
    \begin{tikzpicture}[line cap=round,line join=round,x=1.0cm,y=1.0cm,scale=.7]
\clip(-1.5,1.) rectangle (7.,7.);
\draw [line width=.2pt] (2.,6.)-- (2.,2.);
\draw [line width=.2pt] (2.,2.)-- (2.5,2);
\draw [line width=.2pt] (2.,6.)-- (2.5,6);
\draw [line width=.2pt] (5.5,6.0)-- (6,6.);
\draw [line width=.2pt] (6.,6.)-- (6.,2.);
\draw [line width=.2pt] (5.5,2.)-- (6,2);
\draw [line width=.2pt] (4.,6.)-- (4.,2.);
\draw [line width=.2pt] (2.,4.)-- (6.,4.);
\draw (-1.8,4.6) node[anchor=north west,scale=1.3] {$A(B)=$};
\draw (2.5,5.8) node[anchor=north west,scale=1.5] {$0$};
\draw (4.5,5.7) node[anchor=north west,scale=1.5] {$A$};
\draw (2.5,3.8) node[anchor=north west,scale=1.5] {$A^t$};
\draw (4.5,3.8) node[anchor=north west,scale=1.5] {$0$};
\draw (2.6,7) node[anchor=north west,scale=1.] {$X$};
\draw (4.5,7) node[anchor=north west,scale=1.] {$Y$};
\draw (1,5.5) node[anchor=north west,scale=1.] {$X$};
\draw (1,3.5) node[anchor=north west,scale=1.] {$Y$};
\end{tikzpicture}
  
\end{figure}
where $A$ is the biadjacency matrix of $B$ and also it is not of Ferrers dimension 2. Furthermore, if we transform the vertices of either the X-partite set or the Y-partite set into a complete graph, then B evolves into a chordal graph, and the adjacency matrices of these chordal graphs take one of the following forms:

\begin{figure}[H]
    \centering
    \begin{tikzpicture}[line cap=round,line join=round,x=1.0cm,y=1.0cm,scale=.8]
\clip(-1.5,1.) rectangle (15.,7.);
\draw [line width=.2pt] (2.,6.)-- (2.,2.);
\draw [line width=.2pt] (2.,2.)-- (2.5,2.0);
\draw [line width=.2pt] (2.,6.)-- (2.5,6.0);
\draw [line width=.2pt] (5.5,6.0)-- (6.,6.);
\draw [line width=.2pt] (6.,6.)-- (6.,2.);
\draw [line width=.2pt] (6.,2.)-- (5.5,2.0);
\draw [line width=.2pt] (4.,6.)-- (4.,2.);
\draw [line width=.2pt] (2.,4.)-- (6.,4.);
\draw (-1.2,4.5) node[anchor=north west,scale=1.2] {$A_1(B)=$};
\draw (2.5,5.7) node[anchor=north west,scale=1.5] {$1$};
\draw (4.5,5.7) node[anchor=north west,scale=1.5] {$A$};
\draw (2.3,3.8) node[anchor=north west,scale=1.5] {$A^t$};
\draw (4.6,3.8) node[anchor=north west,scale=1.5] {$0$};
\draw (2.8,6.7) node[anchor=north west] {$X$};
\draw (4.9,6.7) node[anchor=north west] {$Y$};
\draw (1.2,5.3) node[anchor=north west] {$X$};
\draw (1.2,3.3) node[anchor=north west] {$Y$};
\draw [line width=.2pt] (10.,6.)-- (10.5,6);
\draw [line width=.2pt] (10.,6.)-- (10.,2.);
\draw [line width=.2pt] (10.,2.)-- (10.5,2.0);
\draw [line width=.2pt] (13.5,6)-- (14.,6.);
\draw [line width=.2pt] (14.,6.)-- (14.,2.);
\draw [line width=.2pt] (14.,2.)-- (13.5,2);
\draw (10.8,6.6) node[anchor=north west] {$X$};
\draw (12.9,6.6) node[anchor=north west] {$Y$};
\draw (9.2,5.4) node[anchor=north west] {$X$};
\draw (9.2,3.1) node[anchor=north west] {$Y$};
\draw (6.9,4.5) node[anchor=north west,scale=1.2] {$A_2(B)=$};
\draw [line width=.2pt] (10.,4.)-- (14.,4.);
\draw [line width=.2pt] (12.,6.)-- (12.,2.);
\draw (10.6,5.7) node[anchor=north west,scale=1.5] {$0$};
\draw (10.5,3.8) node[anchor=north west,scale=1.5] {$A^t$};
\draw (12.6,5.7) node[anchor=north west,scale=1.5] {$A$};
\draw (12.5,3.8) node[anchor=north west,scale=1.5] {$1$};
\draw (6.,4.4) node[anchor=north west] {or};
\end{tikzpicture}
    
\end{figure}
Now, the graphs corresponding to these two matrices are split graphs, but they are not co-TT graphs (since they do not possess the property of  Ferrers dimension 2 [Theorem 8]). As a result, they serve as forbidden graphs for split co-TT graphs, and consequently, they are also forbidden graphs for co-TT graphs. We denote the forbidden families of split co-TT graph as $\mathcal{\widetilde{S}}$.

The following forbidden induced subgraph (Fig. 8) for split co-TT graphs is derived from the bigraph $C_3$ of Fig. 7 as described earlier .
\begin{figure}[H]
\centering
\begin{tikzpicture}[line cap=round,line join=round,x=1.0cm,y=1.0cm,scale=.5]
\clip(4.5,3.) rectangle (15.3,14.);
\draw [line width=.1pt] (10.,12.)-- (6.,4.);
\draw [line width=.1pt] (10.,12.)-- (14.,4.);
\draw [line width=.1pt] (8.,8.)-- (12.,8.);
\draw [line width=.1pt] (6.,4.)-- (14.,4.);
\draw [line width=.1pt] (8.,8.)-- (10.,4.);
\draw [line width=.1pt] (12.,8.)-- (10.,4.);
\begin{scriptsize}
\draw [fill=black] (10.,12.) circle (4pt);
\draw (9.3,13.2) node[anchor=north west,scale=1.5] {$x_1$};

\draw [fill=black] (6.,4.) circle (4pt);
\draw (4.5,4.5) node[anchor=north west,scale=1.5] {$x_3$};

\draw [fill=black] (14.,4.) circle (4pt);
\draw (14,4.5) node[anchor=north west,scale=1.5] {$x_2$};

\draw [fill=black] (8.,8.) circle (4pt);
\draw (6.6,8.7) node[anchor=north west,scale=1.5] {$y_3$};

\draw [fill=black] (12.,8.) circle (4pt);
\draw (12,8.7) node[anchor=north west,scale=1.5] {$y_1$};

\draw [fill=black] (10.,4.) circle (4pt);
\draw (9.2,4) node[anchor=north west,scale=1.5] {$y_2$};

\end{scriptsize}
\end{tikzpicture}\\
   Fig. 8. The sun graph $S_3$, which is forbidden induced subgraph of split co-TT graph
\end{figure}
In fact, all the sun graphs $\mathrm{S}_k$ $(k \geq 3)$ can be straightforwardly derived from the class of bipartite graphs $C_i (i \geq 3)$ shown in Fig. 7, following the same approach as for $\mathrm{S}_3$. As previously noted, we've established that signed-interval graphs are equivalent to co-TT graphs. \\
Likewise, the following two graphs are forbidden induced subgraphs for split co-TT graphs and are derived from the graph $G_1$ depicted in Fig. 7.

\begin{figure}[H]
    \centering
    \begin{tikzpicture}[line cap=round,line join=round,x=1.0cm,y=1.0cm,scale=1.2]
\clip(0.,.6) rectangle (8.,5.);
\draw [line width=.2pt] (4.,4.)-- (4.,3.);
\draw [line width=.2pt] (4.,3.)-- (4.,2.);
\draw [line width=.2pt] (4.,2.)-- (4.,1.);
\draw [line width=.2pt] (4.,1.)-- (5.,1.);
\draw [line width=.2pt] (5.,1.)-- (6.,1.);
\draw [line width=.2pt] (6.,1.)-- (7.,1.);
\draw [line width=.2pt] (4.,1.)-- (3.,1.);
\draw [line width=.2pt] (3.,1.)-- (2.,1.);
\draw [line width=.2pt] (2.,1.)-- (1.,1.);
\begin{scriptsize}
\draw [fill=black] (4.,4.) circle (1.5pt);
\draw (4,4.2) node[anchor=north west,scale=1.5] {$x_1$};
\draw [fill=black] (4.,3.) circle (1.5pt);
\draw (4,3.2) node[anchor=north west,scale=1.5] {$y_1$};
\draw [fill=black] (4.,2.) circle (1.5pt);
\draw (4,2.2) node[anchor=north west,scale=1.5] {$x_2$};
\draw [fill=black] (4.,1.) circle (1.5pt);
\draw (4,1.4) node[anchor=north west,scale=1.5] {$y_2$};
\draw [fill=black] (5.,1.) circle (1.5pt);
\draw (4.8,1.) node[anchor=north west,scale=1.5] {$x_3$};
\draw [fill=black] (6.,1.) circle (1.5pt);
\draw (5.8,1) node[anchor=north west,scale=1.5] {$y_3$};
\draw [fill=black] (7.,1.) circle (1.5pt);
\draw (6.8,1) node[anchor=north west,scale=1.5] {$x_5$};
\draw [fill=black] (3.,1.) circle (1.5pt);
\draw (2.8,1) node[anchor=north west,scale=1.5] {$x_4$};
\draw [fill=black] (2.,1.) circle (1.5pt);
\draw (1.8,1) node[anchor=north west,scale=1.5] {$y_4$};
\draw [fill=black] (1.,1.) circle (1.5pt);
\draw (0.8,1) node[anchor=north west,scale=1.5] {$x_6$};

\end{scriptsize}
\end{tikzpicture}\vspace{.5cm}\\
A labeling of the graph $G_1$
  
\end{figure}
\begin{figure}[H]
    \centering
    \begin{tikzpicture}[line cap=round,line join=round,x=1.0cm,y=1.0cm,scale=1.]
\clip(3.,1.5) rectangle (11.,6.5);
\draw [line width=.2pt] (5.,5.)-- (7.,5.);
\draw [line width=.2pt] (7.,5.)-- (8.,4.);
\draw [line width=.2pt] (8.,4.)-- (7.,3.);
\draw [line width=.2pt] (7.,3.)-- (5.,3.);
\draw [line width=.2pt] (5.,3.)-- (4.,4.);
\draw [line width=.2pt] (4.,4.)-- (5.,5.);
\draw [line width=.2pt] (6.,6.)-- (5.,5.);
\draw [line width=.2pt] (6.,6.)-- (7.,5.);
\draw [line width=.2pt] (5.,5.)-- (7.,3.);
\draw [line width=.2pt] (5.,5.)-- (8.,4.);
\draw [line width=.2pt] (5.,5.)-- (5.,3.);
\draw [line width=.2pt] (4.,4.)-- (7.,5.);
\draw [line width=.2pt] (4.,4.)-- (8.,4.);
\draw [line width=.2pt] (4.,4.)-- (7.,3.);
\draw [line width=.2pt] (5.,3.)-- (7.,5.);
\draw [line width=.2pt] (5.,3.)-- (8.,4.);
\draw [line width=.2pt] (7.,5.)-- (7.,3.);
\draw [line width=.2pt] (8.,4.)-- (10.,4.);
\draw [line width=.2pt] (7.,5.)-- (10.,4.);
\draw [line width=.2pt] (7.,3.)-- (10.,4.);
\draw [line width=.2pt] (5.,3.)-- (8.,2.);
\draw [line width=.2pt] (8.,4.)-- (8.,2.);
\draw [line width=.2pt] (4.,4.)-- (5.,2.);
\draw [line width=.2pt] (7.,3.)-- (5.,2.);
\begin{scriptsize}
\draw [fill=black] (5.,5.) circle (1.5pt);
\draw (4.3,5.4) node[anchor=north west,scale=1.5] {$x_1$};
\draw [fill=black] (7.,5.) circle (1.5pt);
\draw (7,5.4) node[anchor=north west,scale=1.5] {$x_2$};
\draw [fill=black] (8.,4.) circle (1.5pt);
\draw (8,4) node[anchor=north west,scale=1.5] {$x_3$};
\draw [fill=black] (7.,3.) circle (1.5pt);
\draw (7,3) node[anchor=north west,scale=1.5] {$x_4$};
\draw [fill=black] (5.,3.) circle (1.5pt);
\draw (4.8,2.9) node[anchor=north west,scale=1.5] {$x_5$};
\draw [fill=black] (4.,4.) circle (1.5pt);
\draw (3.2,4.3) node[anchor=north west,scale=1.5] {$x_6$};
\draw [fill=black] (6.,6.) circle (1.5pt);
\draw (5.8,6.6) node[anchor=north west,scale=1.5] {$y_1$};
\draw [fill=black] (10.,4.) circle (1.5pt);
\draw (10,4.3) node[anchor=north west,scale=1.5] {$y_2$};
\draw [fill=black] (8.,2.) circle (1.5pt);
\draw (7.8,2) node[anchor=north west,scale=1.5] {$y_3$};
\draw [fill=black] (5.,2.) circle (1.5pt);
\draw (4.7,2) node[anchor=north west,scale=1.5] {$y_4$};
\end{scriptsize}
\end{tikzpicture}
  
\end{figure}
\begin{figure}[H]
    \centering
    \begin{tikzpicture}[line cap=round,line join=round,x=1.0cm,y=1.0cm,scale=.6]
\clip(2.,1) rectangle (16.,13.);
\draw [line width=.2pt] (6.,10.)-- (10.,10.);
\draw [line width=.2pt] (6.,10.)-- (6.,6.);
\draw [line width=.2pt] (6.,6.)-- (10.,6.);
\draw [line width=.2pt] (10.,10.)-- (10.,6.);
\draw [line width=.2pt] (6.,10.)-- (10.,6.);
\draw [line width=.2pt] (10.,10.)-- (6.,6.);
\draw [line width=.2pt] (6.,10.)-- (4.,8.);
\draw [line width=.2pt] (4.,8.)-- (6.,6.);
\draw [line width=.2pt] (6.,10.)-- (8.,12.);
\draw [line width=.2pt] (8.,12.)-- (10.,10.);
\draw [line width=.2pt] (6.,6.)-- (6.,4.);
\draw [line width=.2pt] (10.,6.)-- (10.,4.);
\draw [line width=.2pt] (10.,6.)-- (12.,6.);
\draw [line width=.2pt] (6.,10.)-- (12.,6.);
\draw [line width=.2pt] (10.,10.)-- (12.,10.);
\begin{scriptsize}
\draw [fill=black] (6.,10.) circle (2.5pt);
\draw (5.1,10.9) node[anchor=north west,scale=1.5] {$y_2$};
\draw [fill=black] (10.,10.) circle (2.5pt);
\draw (9.8,10.9) node[anchor=north west,scale=1.5] {$y_3$};
\draw [fill=black] (6.,6.) circle (2.5pt);
\draw (5.,6) node[anchor=north west,scale=1.5] {$y_1$};
\draw [fill=black] (10.,6.) circle (2.5pt);
\draw (9.9,6) node[anchor=north west,scale=1.5] {$y_4$};
\draw [fill=black] (4.,8.) circle (2.5pt);
\draw (2.7,8.5) node[anchor=north west,scale=1.5] {$x_2$};
\draw [fill=black] (8.,12.) circle (2.5pt);
\draw (7.5,12.8) node[anchor=north west,scale=1.5] {$x_3$};
\draw [fill=black] (6.,4.) circle (2.5pt);
\draw (5.5,4) node[anchor=north west,scale=1.5] {$x_1$};
\draw [fill=black] (10.,4.) circle (2.5pt);
\draw (9.5,4) node[anchor=north west,scale=1.5] {$x_6$};
\draw [fill=black] (12.,6.) circle (2.5pt);
\draw (12,6.4) node[anchor=north west,scale=1.5] {$x_4$};
\draw [fill=black] (12.,10.) circle (2.5pt);
\draw (12,10.4) node[anchor=north west,scale=1.5] {$x_5$};

\end{scriptsize}
\end{tikzpicture}\\
Fig. 9. Two forbidden induced subgraphs for split co-TT graphs are obtained from the graph $G_1$.
  
\end{figure}

\par We have mentioned earlier that, it has been deduced in \cite{24} that every co-TT graph is strongly chordal. Consequently, every co-TT graph is also chordal, and the sun graphs $\mathrm{S}_k (k \geq 3)$ are indeed forbidden induced subgraphs for co-TT graphs. Additionally, it can be verified that if we remove any vertex from $\mathrm{S}_k$ $(k \geq 4)$, then the resulting graph becomes a co-TT graph.\\

Now, we will demonstrate that the graphs  $\mathrm{T_n}$ and $\mathrm{H_n}$ depicted in Fig. 4 are indeed co-TT graphs but the graphs $\mathrm{W}$ and $\mathrm{T}$ are not  co-TT graphs. As a reminder, for an interval $I_v$, $l(v)$ and $r(v)$ respectively represent the left endpoint and right endpoint of $I_v$.\vspace{.3cm} \\
\begin{figure}[H]
    \centering
    \begin{tikzpicture}[line cap=round,line join=round,x=1.0cm,y=1.0cm,scale=1.4]
\clip(0.,0.) rectangle (6.,4.);
\draw [line width=.2pt] (1.,2.)-- (2.,2.);
\draw [line width=.2pt] (2.,2.)-- (3.,2.);
\draw [line width=.2pt] (3.,2.)-- (4.,2.);
\draw [line width=.2pt] (4.,2.)-- (5.,2.);
\draw [line width=.2pt] (1.,2.)-- (3.,1.);
\draw [line width=.2pt] (2.,2.)-- (3.,1.);
\draw [line width=.2pt] (3.,2.)-- (3.,1.);
\draw [line width=.2pt] (4.,2.)-- (3.,1.);
\draw [line width=.2pt] (5.,2.)-- (3.,1.);
\draw [line width=.2pt] (3.,3.)-- (3.,2.);
\draw (2.7,3.5) node[anchor=north west,scale=1.5] {$y$};
\draw (0.7,2.5) node[anchor=north west,scale=1.5] {$x_1$};
\draw (1.8,2.5) node[anchor=north west,scale=1.5] {$x_2$};
\draw (3.,2.5) node[anchor=north west,scale=1.5] {$x_3$};
\draw (3.8,2.5) node[anchor=north west,scale=1.5] {$x_4$};
\draw (4.8,2.5) node[anchor=north west,scale=1.5] {$x_5$};
\draw (2.8,1) node[anchor=north west,scale=1.5] {$x$};
\begin{scriptsize}
\draw [fill=black] (1.,2.) circle (1.5pt);
\draw [fill=black] (2.,2.) circle (1.5pt);
\draw [fill=black] (3.,2.) circle (1.5pt);
\draw [fill=black] (4.,2.) circle (1.5pt);
\draw [fill=black] (5.,2.) circle (1.5pt);
\draw [fill=black] (3.,1.) circle (1.5pt);
\draw [fill=black] (3.,3.) circle (1.5pt);
\end{scriptsize}
\end{tikzpicture}\\
    Fig. 10. A labeling of the graph $W$.
\end{figure}

\begin{lem}\label{l12}
    The graph $W_1$ obtained after removing the edges $xx_1$ and $xx_5$ from $W$ (Fig. 10), is a co-TT graph.
\end{lem}
\begin{figure}[H]
    \centering
    \begin{tikzpicture}[line cap=round,line join=round,x=1.0cm,y=1.0cm,scale=1.7]
\clip(0.,0.) rectangle (6.,4.);
\draw [line width=.2pt] (1.,2.)-- (2.,2.);
\draw [line width=.2pt] (2.,2.)-- (3.,2.);
\draw [line width=.2pt] (3.,2.)-- (4.,2.);
\draw [line width=.2pt] (4.,2.)-- (5.,2.);
\draw [line width=.2pt] (2.,2.)-- (3.,1.);
\draw [line width=.2pt] (3.,2.)-- (3.,1.);
\draw [line width=.2pt] (4.,2.)-- (3.,1.);
\draw [line width=.2pt] (3.,3.)-- (3.,2.);
\draw (2.8,3.5) node[anchor=north west,scale=1.5] {$y$};
\draw (.8,2.4) node[anchor=north west,scale=1.5] {$x_1$};
\draw (1.8,2.4) node[anchor=north west,scale=1.5] {$x_2$};
\draw (3.,2.4) node[anchor=north west,scale=1.5] {$x_3$};
\draw (3.8,2.4) node[anchor=north west,scale=1.5] {$x_4$};
\draw (4.8,2.4) node[anchor=north west,scale=1.5] {$x_5$};
\draw (2.8,1.) node[anchor=north west,scale=1.5] {$x$};
\begin{scriptsize}
\draw [fill=black] (1.,2.) circle (1.5pt);
\draw [fill=black] (2.,2.) circle (1.5pt);
\draw [fill=black] (3.,2.) circle (1.5pt);
\draw [fill=black] (4.,2.) circle (1.5pt);
\draw [fill=black] (5.,2.) circle (1.5pt);
\draw [fill=black] (3.,1.) circle (1.5pt);
\draw [fill=black] (3.,3.) circle (1.5pt);
\end{scriptsize}
\end{tikzpicture}\\
Fig. 11. A labeling of the graph $W_1$.
   
\end{figure}
\begin{proof}
Consider the interval representation of the path $ P: x_1, x_2, x_3, x_4, x_5$. In this representation without loss of generality we assume $l(x_1)< l(x_2)< l(x_3) < l(x_4)< l(x_5)$ and $r(x_1)< r(x_2)<r(x_3)<r(x_4)<r(x_5)$.
Also for two non consecutive vertices $x_r$ and $x_s $( where $r<s), r(x_r)<l(x_s)$. Next, we consider the following possibility. \\
Let $I(x)$ intersects $I(x_2), I(x_3)$ and $I(x_4)$ but not $I(x_1)$ and $I(x_5)$. Now we consider two cases either $l(x_3)<l(x)$ or $r(x)<r(x_3)$. In any case we take a negative interval $I^-(y)$, such that $l(x_3)=r(y)$ and $r(x_3)=l(y)$. This gives a co-TT representation for $W_1$.
\end{proof}
\begin{figure}[H]
    \centering
    \begin{tikzpicture}[line cap=round,line join=round,x=1.0cm,y=1.0cm,scale=1.]
\clip(0.8,0.6) rectangle (13.,6.);
\draw [line width=.2pt] (1.,4.)-- (4.,4.);
\draw [line width=.2pt] (3.,3.)-- (6.,3.);
\draw [line width=.2pt] (5.,4.)-- (8.,4.);
\draw [line width=.2pt] (7.,3.)-- (10.,3.);
\draw [line width=.2pt] (9.,4.)-- (12.,4.);
\draw [line width=.2pt] (5.5,2)-- (7.5,2);
\draw [->,line width=1.5pt] (8.,5.) -- (6.5,5);
\draw [line width=1.5pt] (6.5,5)-- (5.,5.);

\draw (2,4.7) node[anchor=north west,scale=1.5] {$x_1$};
\draw (4.2,3.7) node[anchor=north west,scale=1.5] {$x_2$};
\draw (6.2,4.7) node[anchor=north west,scale=1.5] {$x_3$};
\draw (8.2,3.7) node[anchor=north west,scale=1.5] {$x_4$};
\draw (10.2,4.7) node[anchor=north west,scale=1.5] {$x_5$};
\draw (6.2,2.7) node[anchor=north west,scale=1.5] {$x$};
\draw (6.3,5.8) node[anchor=north west,scale=1.5] {$y$};
\end{tikzpicture}\\
Fig. 12. A co-TT representation of $W_1$.
\end{figure}

\begin{obs}\label{obs12}
\end{obs}
   Based on the proof presented above, it is evident that the graphs $W_1 + xx_5$ and $W_1 + xx_1$ are also co-TT graphs. 

\begin{lem}\label{l13}
    The graph $W$ is not a co-TT graph.
\end{lem}
\begin{proof}
    Based on the proof provided in \cref{l12}, it becomes clear that if $I(x)$ intersects all the intervals ${I(x_1), I(x_2), I(x_3), I(x_4)}$, and $I(x_5)$, then it is not possible to have a negative interval $I^{-}(y)$ for the vertex $y$ such that $I^-(y)$ contained in $I(x_3)$ only. Consequently, $W$ is not a co-TT graph because it lacks a signed-interval representation. 
\end{proof}
\vspace{1cm}

\begin{lem}
    The graph $T$ is not a co-TT graph.
\end{lem}
 \begin{figure}[H]
    \centering
    \begin{tikzpicture}[line cap=round,line join=round,x=1.0cm,y=1.0cm,scale=1]
\clip(1.,1.) rectangle (11.,7.);
\draw [line width=.2pt] (2.,2.)-- (4.,2.);
\draw [line width=.2pt] (4.,2.)-- (6.,2.);
\draw [line width=.2pt] (6.,2.)-- (8.,2.);
\draw [line width=.2pt] (8.,2.)-- (10.,2.);
\draw [line width=.2pt] (6.,2.)-- (6.,4.);
\draw [line width=.2pt] (6.,4.)-- (6.,6.);
\begin{scriptsize}
\draw [fill=black] (2.,2.) circle (2pt);
\draw [fill=black] (4.,2.) circle (2pt);
\draw [fill=black] (6.,2.) circle (2pt);
\draw [fill=black] (8.,2.) circle (2pt);
\draw [fill=black] (10.,2.) circle (2pt);
\draw [fill=black] (6.,4.) circle (2pt);
\draw [fill=black] (6.,6.) circle (2pt);
\draw (1.7,2.7) node[anchor=north west,scale=1.5] {$x_1$};
\draw (3.7,2.7) node[anchor=north west,scale=1.5] {$x_2$};
\draw (6,2.7) node[anchor=north west,scale=1.5] {$x_3$};
\draw (7.7,2.7) node[anchor=north west,scale=1.5] {$x_4$};
\draw (9.7,2.7) node[anchor=north west,scale=1.5] {$x_5$};
\draw (6.,4.5) node[anchor=north west,scale=1.5] {$x_6$};
\draw (6,6.6) node[anchor=north west,scale=1.5] {$x_7$};
\draw (5.6,1.8) node[anchor=north west,scale=2] {$T$};
\end{scriptsize}
\end{tikzpicture}\\
   Fig. 13. A Labeling of the vertices of $T$.
\end{figure}
\begin{proof}
    We shall show that the graph $T$ has no signed-interval representation. First, we consider an interval representation of the path: $x_1, x_2, x_3, x_4, x_5$. In this representation as in Lemma 13, we assume  $l(x_1)< l(x_2)< l(x_3) < l(x_4)< l(x_5)$ and $r(x_1)< r(x_2)<r(x_3)<r(x_4)<r(x_5)$. Also for any two non-consecutive vertices $x_i$ and $x_j (i<j)$ of the path, $r(x_i)< l(x_j)$. Next, we can consider an interval $I(x_6)$ for the vertex $x_6$ such that $r(x_2)< l(x_6)$ and $r(x_6)<l(x_4)$ as $x_6$ is adjacent to $x_3$ only. Now, there exists no positive or negative interval for $x_7$, so that the vertex $x_7$ is adjacent to $x_6$ only. Thus $\mathrm{T}$ has no signed-interval representation. Hence $\mathrm{T}$ is not a co-TT graph. From the symmetry of $T$, the proof is complete.
\end{proof}
\begin{lem}\label{lem14}
    The graphs $T_n$ $(n\geq 1)$ are co-TT graphs.
\end{lem}

\begin{figure}[H]
    \centering
   \begin{tikzpicture}[line cap=round,line join=round,x=1.0cm,y=1.0cm,scale=1.7]
\clip(-1.,0.7) rectangle (7.,6.);
\draw [line width=.2pt] (3.,4.)-- (3.,5.);
\draw [line width=.2pt] (3.,4.)-- (1.,2.);
\draw [line width=.2pt] (1.,2.)-- (0.,2.);
\draw [line width=.2pt] (3.,4.)-- (2.,2.);
\draw [line width=.2pt] (3.,4.)-- (3.,2.);
\draw [line width=.2pt] (1.,2.)-- (2.,2.);
\draw [line width=.2pt] (2.,2.)-- (3.,2.);
\draw [line width=.2pt] (3.,4.)-- (5.,2.);
\draw [line width=.2pt] (5.,2.)-- (6.,2.);
\draw [line width=.2pt] (3.,4.)-- (3.5,2.);
\draw [line width=.2pt] (3.,2.)-- (3.5,2.);
\draw [line width=.2pt] (3.,4.)-- (4.5,2.);
\draw [line width=.2pt] (4.5,2.)-- (5.,2.);
\begin{scriptsize}
\draw [fill=black] (3.,4.) circle (1.5pt);
\draw [fill=black] (3.,5.) circle (1.5pt);
\draw [fill=black] (1.,2.) circle (1.5pt);
\draw [fill=black] (0.,2.) circle (1.5pt);
\draw [fill=black] (2.,2.) circle (1.5pt);
\draw [fill=black] (3.,2.) circle (1.5pt);
\draw [fill=black] (5.,2.) circle (1.5pt);
\draw [fill=black] (6.,2.) circle (1.5pt);
\draw [fill=black] (3.5,2.) circle (1.5pt);
\draw [fill=black] (4.5,2.) circle (1.5pt);
\draw [fill=black] (3.8,2.0) circle (.5pt);
\draw [fill=black] (4.,2.) circle (.5pt);
\draw [fill=black] (4.2,2.0) circle (.5pt);
\draw (-.2,1.9) node[anchor=north west,scale=2] {$v$};
\draw (.8,1.9) node[anchor=north west,scale=1.5] {$x_0$};
\draw (1.6,1.9) node[anchor=north west,scale=1.5] {$x_1$};
\draw (2.5,1.9) node[anchor=north west,scale=1.5] {$x_2$};
\draw (3.2,1.9) node[anchor=north west,scale=1.5] {$x_3$};
\draw (3.9,1.9) node[anchor=north west,scale=1.3] {$x_{n-1}$};
\draw (5.,1.9) node[anchor=north west,scale=1.5] {$x_n$};
\draw (6,1.9) node[anchor=north west,scale=1.8] {$w$};
\draw (2.5,4.5) node[anchor=north west,scale=1.8] {$x$};
\draw (2.6,5.5) node[anchor=north west,scale=2] {$u$};

\end{scriptsize}
\end{tikzpicture}\\
Fig. 14. A labelling of $T_n$.
\end{figure}

\begin{proof}
    As before, we consider the interval representation of the path:\\
    $v, x_0, x_1, x_2, x_3,..., x_{n-1}, x_n, w$. In this representation we may assume $l(v)< l(x_o)< l(x_1)< l(x_2)< ...< l(x_{n-1})< l(x_n)< l(w)$, and $r(v)< r(x_0)< r(x_1)< r(x_2)< ...< r(x_{n-1})< r(x_n)< r(w)$. Also $r(v)< l(x_1)$, $r(x_{n-1})< l(w)$. And for two non-consecutive vertices $x_i $ and $x_j$ with $i<j$, we have $r(x_i)< l(x_j)$. Since the vertex $x$ is adjacent to all $x_i (i=0,1,2,...,n)$, we take $I(x)$ such that $r(v)< l(x)< l(x_1)$ and $r(x_{n-1})< r(x)< l(w)$. Now we take negative interval $I^-(u)$ in the following way. Let $x_i$ and $x_{i+1}$ be two intervals ($1\leq i\leq {n-1}$). Also $l(x_i)< r(u)< l(x_{i+1})$ and $r(x_i)< l(u)< r(x_{i+1})$. Then $I^-(u)$ is contained only in $I(x)$. This gives a signed-interval representation i.e. a co-TT representation of $T_n$. Hence $T_n$ are  co-TT graphs.
\end{proof}
\begin{figure}[H]
    \centering
  \begin{tikzpicture}[line cap=round,line join=round,x=1.0cm,y=1.0cm,scale=.8]

\clip(0.,0.) rectangle (20.,5.);
\draw [line width=.2pt] (1.,2.)-- (4.,2.);
\draw (2.,2.9) node[anchor=north west,scale=1.8] {$v$};
\draw [line width=.2pt] (3.,1.)-- (6.,1.);
\draw (3.8,1.8) node[anchor=north west,scale=1.3] {$x_o$};
\draw [line width=.2pt] (5.,2.)-- (8.,2.);
\draw (5.7,2.8) node[anchor=north west,scale=1.3] {$x_1$};
\draw [line width=.2pt] (7.,1.)-- (10.,1.);
\draw (8.2,1.8) node[anchor=north west,scale=1.3] {$x_2$};
\draw [line width=.2pt] (9.,2.)-- (12.,2.);
\draw (10.2,2.8) node[anchor=north west,scale=1.3] {$x_3$};
\draw [line width=.2pt] (4.606236334088112,3.017224100096699)-- (16.,3.);
\draw (10,3.9) node[anchor=north west,scale=1.5] {$x$};
\draw [line width=.2pt] (14.,1.)-- (15.5,1.);
\draw (14,1.8) node[anchor=north west,scale=1.2] {$x_{n-1}$};
\draw [line width=.2pt] (15.,2.)-- (18.,2.);
\draw (16,2.8) node[anchor=north west,scale=1.3] {$x_n$};
\draw [line width=.2pt] (17.5,1)-- (19.,1.);
\draw (17.9,1.8) node[anchor=north west,scale=1.5] {$w$};
\draw [->,line width=1.5pt] (11.,4.0) -- (10.,4.);
\draw [line width=1.5pt] (10.,4.)-- (8.5,4);
\draw (9,4.8) node[anchor=north west,scale=1.5] {$u$};
\draw [line width=1pt,dotted] (12.2,1.5)-- (13.8,1.5);
\end{tikzpicture}
   Fig. 15. A co-TT representation of $T_n$.
\end{figure}

\begin{lem}\label{lem15}
    The graph $T_0$ is not a co-TT graph.
\end{lem}

\begin{figure}[H]
    \centering
    \begin{tikzpicture}[line cap=round,line join=round,x=1.0cm,y=1.0cm,scale=1.5]
\clip(0.,0.) rectangle (8.,5.);
\draw [line width=.2pt] (1.,1.)-- (2.,1.);
\draw [line width=.2pt] (2.,1.)-- (3.,1.);
\draw [line width=.2pt] (3.,1.)-- (4.,2.);
\draw [line width=.2pt] (4.,2.)-- (5.,1.);
\draw [line width=.2pt] (3.,1.)-- (5.,1.);
\draw [line width=.2pt] (5.,1.)-- (6.,1.);
\draw [line width=.2pt] (6.,1.)-- (7.,1.);
\draw [line width=.2pt] (4.,2.)-- (4.,3.);
\draw [line width=.2pt] (4.,3.)-- (4.,4.);
\begin{scriptsize}
\draw [fill=black] (1.,1.) circle (1.5pt);
\draw (.8,1.6) node[anchor=north west,scale=1.8] {$v'$};
\draw [fill=black] (2.,1.) circle (1.5pt);
\draw (1.8,1.5) node[anchor=north west,scale=1.8] {$v$};
\draw [fill=black] (3.,1.) circle (1.5pt);
\draw (2.5,1.5) node[anchor=north west,scale=1.8] {$x_0$};
\draw [fill=black] (4.,2.) circle (1.5pt);
\draw (3.47,2.3) node[anchor=north west,scale=2] {$x$};
\draw [fill=black] (6.,1.) circle (1.5pt);
\draw (4.8,1.5) node[anchor=north west,scale=1.8] {$x_1$};
\draw [fill=black] (7.,1.) circle (1.5pt);
\draw (6.8,1.6) node[anchor=north west,scale=1.8] {$w'$};
\draw [fill=black] (4.,3.) circle (1.5pt);
\draw (3.5,3.2) node[anchor=north west,scale=1.8] {$u$};
\draw [fill=black] (4.,4.) circle (1.5pt);
\draw (3.5,4.4) node[anchor=north west,scale=1.8] {$u'$};
\draw[fill=black] (5.,1.)  circle (1.5pt);
\draw (5.7,1.5) node[anchor=north west,scale=1.8] {$w$};
\end{scriptsize}
\end{tikzpicture}\\
Fig. 16. The graph $T_0$.
\end{figure}

\begin{proof}
    Obviously the graph $T_0$ contains the graph $T_1$( when $n=1$ in the family $T_n$) as an induced subgraph. The co-TT representation of $T_1$ can be found from Fig. 15. Which is given in the Fig. 17. Also from the symmetry of $T_1$ we can infer that this representation is unique up to suitable modification.
    \begin{figure}[H]
        \centering
        \begin{tikzpicture}[line cap=round,line join=round,x=1.0cm,y=1.0cm]
\clip(0.,1.) rectangle (10.,6.);
\draw [line width=.2pt] (0.,3.)-- (3.,3.);
\draw (1.2,3.6) node[anchor=north west,scale=1.5] {$v$};
\draw [line width=.2pt] (1.,2.)-- (5.,2.);
\draw (2.6,2.65) node[anchor=north west,scale=1.5] {$x_0$};
\draw [line width=.2pt] (4.,3.)-- (8.,3.);
\draw (5.2,3.65) node[anchor=north west,scale=1.5] {$x_1$};
\draw [line width=.2pt] (7.,2.)-- (10.,2.);
\draw (7.9,2.6) node[anchor=north west,scale=1.5] {$w$};
\draw [line width=.2pt] (3.5,4.0)-- (6.5,4.0);
\draw (4.5,4.6) node[anchor=north west,scale=1.5] {$x$};
\draw [->,line width=1.5pt] (6.5,5.0) -- (4.4,5.0);
\draw [line width=1.5pt] (4.4,5)-- (3.5,5);
\draw (4.5,5.67) node[anchor=north west,scale=1.5] {$u$};
\end{tikzpicture}\\
Fig. 17. A co-TT representation of $T_1$.
    \end{figure}
    From the Fig. 17 it is clear that we have positive intervals $I(v')$ and $I(w')$ respectively for $v'$ and $w'$, where $I(v')$ intersects $I(v)$ only and $I(w')$ intersects $I(w)$ only. But can't have any positive or negative interval for the vertex $u'$ so that $u'$ is adjacent to $u$ only. So there exists no co-TT representation of $T_0$ and accordingly $T_0$ is not a co-TT graph.
\end{proof}
\begin{obs}\label{obs17}
\end{obs}
    In the graph $T_n (n > 1)$, if we include the edge $uu′$, it will contain $T$ as an induced subgraph.

\begin{lem}
   The graphs $H_n (n\geq 2)$ are co-TT graphs. 
\end{lem}
\begin{figure}[H]
    \centering
    \begin{tikzpicture}[line cap=round,line join=round,x=1.0cm,y=1.0cm,scale=1.7]
\clip(0.,1.) rectangle (6.,6.);
\draw [line width=.2pt] (2.,4.)-- (4.,4.);
\draw [line width=.2pt] (2.,4.)-- (3.,5.);
\draw [line width=.2pt] (3.,5.)-- (4.,4.);
\draw [line width=.2pt] (2.,4.)-- (1.,2.);
\draw [line width=.2pt] (4.,4.)-- (5.,2.);
\draw [line width=.2pt] (4.,4.)-- (4.,2.);
\draw [line width=.2pt] (4.,2.)-- (5.,2.);
\draw [line width=.2pt] (2.,4.)-- (5.,2.);
\draw [line width=.2pt] (2.,4.)-- (4.,2.);
\draw [line width=.2pt] (1.,2.)-- (2.,2.);
\draw [line width=.2pt] (2.,4.)-- (2.,2.);
\draw [line width=.2pt] (2.,4.)-- (2.7,2);
\draw [line width=.2pt] (2.,2.)-- (2.7,2);
\draw [line width=.2pt] (1.,2.)-- (4.,4.);
\draw [line width=.2pt] (2.,2.)-- (4.,4.);
\draw [line width=.2pt] (2.7,2)-- (4.,4.);
\draw [line width=.2pt] (0.5,3.5)-- (2.,4.);
\draw [line width=.2pt] (0.5,3.5)-- (1.,2.);
\draw [line width=.2pt] (4.,4.)-- (5.5,3.5);
\draw [line width=.2pt] (5.5,3.5)-- (5.,2.);
\draw[dotted, line width=1pt] (2.7,2)--(4,2);

\begin{scriptsize}
\draw [fill=black] (2.,4.) circle (1.2pt);
\draw (1.7,4.3) node[anchor=north west,scale=1.5] {$y$};
\draw [fill=black] (4.,4.) circle (1.5pt);
\draw (4.,4.3) node[anchor=north west,scale=1.5] {$z$};
\draw [fill=black] (3.,5.) circle (1.5pt);
\draw (2.8,5.4) node[anchor=north west,scale=1.5] {$u$};
\draw [fill=black] (1.,2.) circle (1.5pt);
\draw (.67,2.) node[anchor=north west,scale=1.5] {$x_1$};
\draw [fill=black] (5.,2.) circle (1.5pt);
\draw (4.8,2) node[anchor=north west,scale=1.5] {$x_n$};
\draw [fill=black] (4.,2.) circle (1.5pt);
\draw (3.8,2) node[anchor=north west,scale=1.5] {$x_{n-1}$};
\draw [fill=black] (2.,2.) circle (1.5pt);
\draw (1.8,2) node[anchor=north west,scale=1.5] {$x_2$};
\draw [fill=black] (2.7,2) circle (1.5pt);
\draw (2.5,2) node[anchor=north west,scale=1.5] {$x_3$};
\draw [fill=black] (0.5,3.5) circle (1.5pt);
\draw (0.1,3.7) node[anchor=north west,scale=1.5] {$w$};
\draw [fill=black] (5.5,3.5) circle (1.5pt);
\draw (5.5,3.7) node[anchor=north west,scale=1.5] {$v$};
\end{scriptsize}
\end{tikzpicture}\\
Fig. 18. A labelling of the graph $H_n$.
   
\end{figure}
\begin{proof}
    Consider the interval representation of the path $w, x_1, x_2, ..., x_{n-1}, x_n, v$. As before, we assume $l(w)< l(x_1)< l(x_2)< ...< l(x_{n-1})< l(x_n)< l(x_v)$ and $r(w)< r(x_1)< r(x_2)< ...< r(x_{n-1})< r(x_n)< r(x_v)$. Also, $r(w)< l(x_2)$, $r(x_{n-1})< l(v)$ and for $x_i$ and $x_{i+2}$ we have $ r(x_i)< l(x_{i+2}) $ $(1\leq i\leq n-2)$. Now, $y$ is adjacent to all the vertices of this path except $v$ and $z$ is adjacent to all the vertices of this path except $w$. So we take $I(y)$ such that $l(w)< l(x_1)< l(y)$ and $l(x_n)<r(y)<l(v)$. Similarly, we take $I(z)$ such that $r(w)< l(z)< l(x_2)$ and $r(z)< r(x_n)< r(v)$. Again $u$ is adjacent to $y$ and $z$ only. Also let $x_i$ and $x_{i+1}$ are two  consecutive vertices $(1< i< n-1)$ then we take a negative interval $I^-(u)$ such that $l(x_i)< r(u)< l(x_{i+1})$ and $r(x_{i})< l(u)< r(x_{i+1})$. Then $I^-(u)$ is contained in both $I(y)$ and $I(z)$ also contained in no other interval. Thus we have a co-TT representation of $H_n (n\geq 2)$.
\end{proof}
\begin{figure}[H]
    \centering
    \begin{tikzpicture}[line cap=round,line join=round,x=1.0cm,y=1.0cm,scale=.8]
\clip(-.1,-.5) rectangle (21.,6.5);
\draw [line width=.2pt] (0.,2.)-- (2.,2.);
\draw (.6,2.8) node[anchor=north west,scale=1.5] {$w$};
\draw [line width=.2pt] (1.,1.)-- (4.,1.);
\draw (1.8,1.8) node[anchor=north west,scale=1.5] {$x_1$};
\draw [line width=.2pt] (3.,2.)-- (6.,2.);
\draw (3.8,2.8) node[anchor=north west,scale=1.5] {$x_2$};
\draw [line width=.2pt] (5.,1.)-- (8.,1.);
\draw (5.8,1.8) node[anchor=north west,scale=1.5] {$x_3$};
\draw [line width=.2pt] (7.,2.)-- (10.,2.);
\draw (7.8,2.8) node[anchor=north west,scale=1.5] {$x_4$};
\draw [line width=.2pt] (9.,1.)-- (12.,1.);
\draw (9.8,1.8) node[anchor=north west,scale=1.5] {$x_5$};
\draw [line width=.2pt] (14.,2.)-- (17.,2.);
\draw (14.8,2.9) node[anchor=north west,scale=1.5] {$x_{n-1}$};
\draw [line width=.2pt] (16.,1.)-- (19.,1.);
\draw (16.8,1.8) node[anchor=north west,scale=1.5] {$x_n$};
\draw [line width=.2pt] (18.,2.)-- (20.,2.);
\draw (18.8,2.8) node[anchor=north west,scale=1.5] {$v$};
\draw [line width=.2pt] (1.5,3.)-- (17.,3.);
\draw (7.8,3.9) node[anchor=north west,scale=1.5] {$y$};
\draw [line width=.2pt] (18.5,4.)-- (3.,4.);
\draw (11.8,4.9) node[anchor=north west,scale=1.7] {$z$};
\draw [->,line width=1pt] (11.,5.) -- (10.,5.);
\draw [line width=1.pt] (10.,5.)-- (9.,5.);
\draw (9.8,5.8) node[anchor=north west,scale=1.5] {$u$};
\draw [line width=1pt,dotted] (12.2,1.5)-- (13.8,1.5);

\end{tikzpicture}\\
Fig. 19. A co-TT representation of $H_n$.
\end{figure}

\begin{obs}\label{obs19}
\end{obs}
    The graph $H_1$ coincide with the graph $\mathrm{S_3}$ (the 3-sun). Also $H^*_n (n\geq 2)$, the graphs obtained from $H_n$ by adding a new vertex $u'$ and joining it with $u$ are forbidden induced subgraphs for co-TT graphs, which we shall state in the next lemma.

\begin{lem}
    The graphs $H^*_n (n\geq 2)$ are forbidden induced subgraphs of co-TT graphs.
\end{lem}
\begin{proof}
    The proof follows directly from the proof of lemma 18.
\end{proof}

\begin{figure}[H]
    \centering
    \begin{tikzpicture}[line cap=round,line join=round,x=1.0cm,y=1.0cm,scale=1.7]
\clip(0.,1.) rectangle (6.,7.);
\draw [line width=.2pt] (2.,4.)-- (4.,4.);
\draw [line width=.2pt] (2.,4.)-- (3.,5.);
\draw [line width=.2pt] (3.,5.)-- (4.,4.);
\draw [line width=.2pt] (2.,4.)-- (1.,2.);
\draw [line width=.2pt] (4.,4.)-- (5.,2.);
\draw [line width=.2pt] (4.,4.)-- (4.,2.);
\draw [line width=.2pt] (4.,2.)-- (5.,2.);
\draw [line width=.2pt] (2.,4.)-- (5.,2.);
\draw [line width=.2pt] (2.,4.)-- (4.,2.);
\draw [line width=.2pt] (1.,2.)-- (2.,2.);
\draw [line width=.2pt] (2.,4.)-- (2.,2.);
\draw [line width=.2pt] (2.,4.)-- (2.7,2);
\draw [line width=.2pt] (2.,2.)-- (2.7,2);
\draw [line width=.2pt] (1.,2.)-- (4.,4.);
\draw [line width=.2pt] (2.,2.)-- (4.,4.);
\draw [line width=.2pt] (2.7,2)-- (4.,4.);
\draw [line width=.2pt] (0.5,3.5)-- (2.,4.);
\draw [line width=.2pt] (0.5,3.5)-- (1.,2.);
\draw [line width=.2pt] (4.,4.)-- (5.5,3.5);
\draw [line width=.2pt] (5.5,3.5)-- (5.,2.);
\draw[dotted, line width=1pt] (2.7,2)--(4,2);
\begin{scriptsize}
\draw [fill=black] (2.,4.) circle (1.5pt);
\draw (1.7,4.3) node[anchor=north west,scale=1.5] {$y$};
\draw [fill=black] (4.,4.) circle (1.5pt);
\draw (4.,4.3) node[anchor=north west,scale=1.5] {$z$};
\draw [fill=black] (3.,5.) circle (1.5pt);
\draw (3,5.2) node[anchor=north west,scale=1.5] {$u$};
\draw [fill=black] (1.,2.) circle (1.5pt);
\draw (.67,2.) node[anchor=north west,scale=1.5] {$x_1$};
\draw [fill=black] (5.,2.) circle (1.5pt);
\draw (4.8,2) node[anchor=north west,scale=1.5] {$x_n$};
\draw [fill=black] (4.,2.) circle (1.5pt);
\draw (3.8,2) node[anchor=north west,scale=1.5] {$x_{n-1}$};
\draw [fill=black] (2.,2.) circle (1.5pt);
\draw (1.8,2) node[anchor=north west,scale=1.5] {$x_2$};
\draw [fill=black] (2.7,2) circle (1.5pt);
\draw (2.5,2) node[anchor=north west,scale=1.5] {$x_3$};
\draw [fill=black] (0.5,3.5) circle (1.5pt);
\draw (0.1,3.7) node[anchor=north west,scale=1.5] {$w$};
\draw [fill=black] (5.5,3.5) circle (1.5pt);
\draw (5.5,3.7) node[anchor=north west,scale=1.5] {$v$};
\draw[line width=.2pt] (3,5)--(3,6);
\draw[fill=black] (3,6) circle (1.5pt);
\draw (3,6.2) node[anchor=north west,scale=1.5] {$u'$};

\end{scriptsize}
\end{tikzpicture}\\
Fig. 20. The graph $H^*_n (n\geq 2)$.
   
\end{figure}\vspace{1cm}
\begin{obs}\label{obs21}
\end{obs}
    Again, we consider the graphs $H_n$, where $n\leq 3$. First, we consider the graph $H_3$. If we take a vertex $x'$, and join it with the vertex $x_1$ or the vertex $x_3$ then in any case from Fig. 19 we can conclude that it will be a co-TT graph by taking a negative interval $I^-(x')$, where $r(x')$= $l(x_1)$ and $l(x')$= $r(x_1)$ or $r(x')$= $l(x_3)$ and $l(x')$= $r(x_3)$. But, if we join $x'$ with $x_2$ then the resulting graph (Fig. 21)
 is a forbidden induced graph for co-TT graph. Because now we have no positive or negative interval for the vertex $x'$. It is easy to check that $\mathrm{H_2}$ is a co-TT graph and $\mathrm{H_1}$= $\mathrm{S_3}$. Next, consider the graphs $H_n$, where $n>3$. Similarly, if we join a vertex $x'$ with $x_1$ or $x_n$ then we have co-TT graph (as now we can take a negative interval $I^-(x')$ as before). Finally, if we join $x'$ with any of the vertices $x_2, x_3,..., x_{n-1}$ then every resulting graph will contain $W$ as an induced subgraph.

 \begin{figure}[H]
    \centering
    \begin{tikzpicture}[line cap=round,line join=round,x=1.0cm,y=1.0cm,scale=1.5]
\clip(0.,0.) rectangle (6.,6.);
\draw [line width=.2pt] (2.,4.)-- (4.,4.);
\draw [line width=.2pt] (2.,4.)-- (3.,5.);
\draw [line width=.2pt] (3.,5.)-- (4.,4.);
\draw [line width=.2pt] (2.,4.)-- (2.,2.);
\draw [line width=.2pt] (2.,2.)-- (3,2);
\draw [line width=.2pt] (2.,4.)-- (3,2);
\draw [line width=.2pt] (2.,4.)-- (4.,2.);
\draw [line width=.2pt] (3,2)-- (4.,2.);
\draw [line width=.2pt] (2.,2.)-- (4.,4.);
\draw [line width=.2pt] (3,2)-- (4.,4.);
\draw [line width=.2pt] (4.,2.)-- (4.,4.);
\draw [line width=.2pt] (1,3.0)-- (2.,4.);
\draw [line width=.2pt] (1,3.0)-- (2.,2.);
\draw [line width=.2pt] (4.,4.)-- (5.,3.);
\draw [line width=.2pt] (4.,2.)-- (5.,3.);
\draw [line width=.2pt] (3,2)-- (3,0.7);
\begin{scriptsize}
\draw [fill=black] (2.,4.) circle (1.5pt);
\draw (1.7,4.4) node[anchor=north west,scale=1.5] {$y$};
\draw [fill=black] (4.,4.) circle (1.5pt);
\draw (3.9,4.4) node[anchor=north west,scale=1.5] {$z$};
\draw [fill=black] (3.,5.) circle (1.5pt);
\draw (2.8,5.4) node[anchor=north west,scale=1.5] {$u$};
\draw [fill=black] (2.,2.) circle (1.5pt);
\draw (1.8,2.) node[anchor=north west,scale=1.5] {$x_1$};
\draw [fill=black] (3,2) circle (1.5pt);
\draw (2.95,2) node[anchor=north west,scale=1.5] {$x_2$};
\draw [fill=black] (4.,2.) circle (1.5pt);
\draw (3.8,2) node[anchor=north west,scale=1.5] {$x_3$};
\draw [fill=black] (1,3.0) circle (1.5pt);
\draw (.55,3.2) node[anchor=north west,scale=1.5] {$w$};
\draw [fill=black] (5.,3.) circle (1.5pt);
\draw (5,3.2) node[anchor=north west,scale=1.5] {$v$};
\draw [fill=black] (3,0.7) circle (1.5pt);
\draw (2.9,.8) node[anchor=north west,scale=1.5] {$x'$};

\end{scriptsize}
\end{tikzpicture}\\
Fig. 21. The graph $H'_3$, a forbidden graph for co-TT graph constructed from $H_3$.
   
\end{figure}

 \par We are now ready to present the main result of this paper. As previously mentioned, the forbidden graphs for split co-TT graphs also serve as forbidden graphs for co-TT graphs. Assume $\mathcal{\widetilde{S}}$ represent the forbidden family of split co-TT graphs, and if we denote $S$ as the set of sun graphs, i.e., $S = \{S_k, k \geq 3\}$, then we have $ S \subset \mathcal{\widetilde{S}}$.

 \par The forthcoming figure (Fig. 22) will facilitate the proof of the main result.

\begin{figure}[H]
    \centering
    \begin{tikzpicture}[line cap=round,line join=round,x=1.0cm,y=1.0cm,scale=1.2]
\clip(-6,-7.) rectangle (8,8.);
\draw [line width=0.2pt,color=black] (0.,0.) circle (5.2cm);
\draw [line width=0.2pt,color=red] (0.,0.) circle (3.59cm);
\draw [line width=0.2pt] (1.5215423925667793,-3.5039907084785127)-- (1.6729599335951844,-3.7335779756192693);
\draw [line width=.2pt] (1.6729599335951844,-3.7335779756192693)-- (1.8147272699071277,-3.9417765753278435);
\draw [line width=0.2pt] (1.8147272699071277,-3.9417765753278435)-- (2.2297982210980987,-4.0638562668546);
\draw [line width=0.2pt] (2.2297982210980987,-4.0638562668546)-- (2.5570730597014375,-4.158936901586497);
\draw [line width=0.2pt] (1.8147272699071277,-3.9417765753278435)-- (1.7252021627875067,-4.324292942111679);
\draw [line width=0.2pt] (1.7252021627875067,-4.324292942111679)-- (1.643815701769669,-4.649838786183029);
\draw [line width=0.2pt,color=black] (0.,0.) circle (2.55cm);
\draw [line width=0.2pt] (-0.9006477714987563,-2.9609558690666966)-- (-0.3499055979930673,-2.9554107057257113);
\draw [line width=0.2pt] (-0.3499055979930673,-2.9554107057257113)-- (0.22999769093414124,-2.9435613234908056);
\draw [line width=0.2pt] (0.22999769093414124,-2.9435613234908056)-- (0.7980732349818327,-2.936894918097084);
\draw [line width=0.2pt] (0.7980732349818327,-2.936894918097084)-- (1.412826790094711,-2.908772232339024);
\draw [line width=0.2pt] (0.22999769093414124,-2.9435613234908056)-- (0.20556803086575529,-2.6036107407817894);
\draw [line width=0.2pt] (0.2397593713406453,-3.4402211410511803)-- (0.7980732349818327,-2.936894918097084);
\draw [line width=0.2pt] (-0.3499055979930673,-2.9554107057257113)-- (0.2397593713406453,-3.4402211410511803);
\draw [line width=0.2pt] (0.22999769093414124,-2.9435613234908056)-- (0.2397593713406453,-3.4402211410511803);
\draw [line width=0.2pt,color=black] (0.,0.) circle (1.3cm);
\draw [rotate around={0.:(2.9478171617119764,0.)},line width=0.2pt,color=black] (2.9478171617119764,0.) ellipse (4.810375078109238cm and 2.2562480518700077cm);
\draw [color=black](-2.2022023231098578,4.5) node[anchor=north west,scale=1.] {Strongly Chordal graph};
\draw [color=red](-1.2392432466268948,3.11935748287192) node[anchor=north west,scale=1.] {Co-TT graph};
\draw [color=black](-0.8,0.6) node[anchor=north west,scale=1]{\parbox{2.5204616997099416 cm}{Threshold  graph}};
\draw [line width=0.2pt] (-1.166879479858047,-1.8932440729743507)-- (-0.7760034911473422,-1.9053264742592335);
\draw [line width=0.2pt] (-0.7760034911473422,-1.9053264742592335)-- (-0.41024385994328416,-1.9129464665759848);
\draw [line width=0.2pt] (-1.166879479858047,-1.8932440729743507)-- (-1.184343116153917,-2.1472606009142776);
\draw [line width=0.2pt] (-1.184343116153917,-2.1472606009142776)-- (-0.7912434757808446,-2.164406213028776);
\draw [line width=0.2pt] (-0.7912434757808446,-2.164406213028776)-- (-0.4095927059371437,-2.177425063607144);
\draw [line width=0.2pt] (-0.4095927059371437,-2.177425063607144)-- (-0.41024385994328416,-1.9129464665759848);
\draw [line width=0.2pt] (-0.7760034911473422,-1.9053264742592335)-- (-0.7912434757808446,-2.164406213028776);
\draw [line width=0.2pt] (-0.41024385994328416,-1.9129464665759848)-- (-0.7912434757808446,-2.164406213028776);
\draw [line width=0.2pt] (-0.7760034911473422,-1.9053264742592335)-- (-1.184343116153917,-2.1472606009142776);
\draw (1.320200614551507,0.7373008199930228) node[anchor=north west,scale=1.5] {$P_4$};
\draw [line width=0.2pt] (2.9791731925306792,0.6563732232753479)-- (2.8127721831061807,0.3363712820743883);
\draw [line width=0.2pt] (2.9791731925306792,0.6563732232753479)-- (3.2095745901953703,0.34917135972242674);
\draw [line width=0.2pt] (2.8127721831061807,0.3363712820743883)-- (3.2095745901953703,0.34917135972242674);
\draw [line width=0.2pt] (3.2095745901953703,0.34917135972242674)-- (3.439975987860061,0.11876996205773586);
\draw [line width=0.2pt] (2.8127721831061807,0.3363712820743883)-- (2.6079709407375664,0.016369340873428805);
\draw [line width=0.2pt] (2.9791731925306792,0.6563732232753479)-- (2.9407729595865644,1.19397648449296);
\draw (6.0843139403093245,0.6866187633360249) node[anchor=north west,scale=1.5] {$S_k$};
\draw [color=black](5.1,-0.6866187633360249) node[anchor=north west,scale=1.] {Split graph};
\draw [color=black](-1.3,2) node[anchor=north west,scale=1.] {Interval graph};
\draw [line width=0.2pt] (4.121633276408737,0.5697829072153996)-- (4.628575928217337,0.5987510587473195);
\draw [line width=0.2pt] (4.121633276408737,0.5697829072153996)-- (4.150601427940657,0.04835617964084034);
\draw [line width=0.2pt] (4.150601427940657,0.04835617964084034)-- (4.672028155515218,0.07732433117276029);
\draw [line width=0.2pt] (4.628575928217337,0.5987510587473195)-- (4.672028155515218,0.07732433117276029);
\draw [line width=0.2pt] (4.121633276408737,0.5697829072153996)-- (4.3702917985587355,0.9022006773980903);
\draw [line width=0.2pt] (4.3702917985587355,0.9022006773980903)-- (4.628575928217337,0.5987510587473195);
\draw [line width=0.2pt] (4.121633276408737,0.5697829072153996)-- (3.828890553725365,0.2965653865675416);
\draw [line width=0.2pt] (3.828890553725365,0.2965653865675416)-- (4.150601427940657,0.04835617964084034);
\draw [line width=0.2pt] (4.628575928217337,0.5987510587473195)-- (4.150601427940657,0.04835617964084034);
\draw [line width=0.2pt] (4.121633276408737,0.5697829072153996)-- (4.672028155515218,0.07732433117276029);
\draw [line width=0.2pt] (4.628575928217337,0.5987510587473195)-- (4.970848454734126,0.6162432352297802);
\draw [line width=0.2pt] (4.672028155515218,0.07732433117276029)-- (5.017127941065096,0.10716888558911593);
\draw [line width=0.2pt] (4.121633276408737,0.5697829072153996)-- (5.017127941065096,0.10716888558911593);
\draw [line width=0.2pt] (4.150601427940657,0.04835617964084034)-- (4.169072116515638,-0.28401252321621184);
\draw [line width=0.2pt] (4.672028155515218,0.07732433117276029)-- (4.690905902865687,-0.27814922224598665);
\begin{scriptsize}
\draw [fill=black] (1.5215423925667793,-3.5039907084785127) circle (1.0pt);
\draw [fill=black] (1.6729599335951844,-3.7335779756192693) circle (1.0pt);
\draw [fill=black] (1.8147272699071277,-3.9417765753278435) circle (1.0pt);
\draw [fill=black] (2.2297982210980987,-4.0638562668546) circle (1.0pt);
\draw [fill=black] (2.5570730597014375,-4.158936901586497) circle (1.0pt);
\draw [fill=black] (1.7252021627875067,-4.324292942111679) circle (1.0pt);
\draw [fill=black] (1.643815701769669,-4.649838786183029) circle (1.0pt);
\draw [fill=black] (-0.9006477714987563,-2.9609558690666966) circle (1.0pt);
\draw [fill=black] (-0.3499055979930673,-2.9554107057257113) circle (1.0pt);
\draw [fill=black] (0.22999769093414124,-2.9435613234908056) circle (1.0pt);
\draw [fill=black] (0.7980732349818327,-2.936894918097084) circle (1.0pt);
\draw [fill=black] (1.412826790094711,-2.908772232339024) circle (1.0pt);
\draw [fill=black] (0.20556803086575529,-2.6036107407817894) circle (1.0pt);
\draw [fill=black] (0.2397593713406453,-3.4402211410511803) circle (1.0pt);
\draw [fill=black] (-1.166879479858047,-1.8932440729743507) circle (1.0pt);
\draw [fill=black] (-0.7760034911473422,-1.9053264742592335) circle (1.0pt);
\draw [fill=black] (-0.41024385994328416,-1.9129464665759848) circle (1.0pt);
\draw [fill=black] (-1.184343116153917,-2.1472606009142776) circle (1.0pt);
\draw [fill=black] (-0.7912434757808446,-2.164406213028776) circle (1.0pt);
\draw [fill=black] (-0.4095927059371437,-2.177425063607144) circle (1.0pt);
\draw [fill=black] (2.9791731925306792,0.6563732232753479) circle (1.0pt);
\draw [fill=black] (2.8127721831061807,0.3363712820743883) circle (1.0pt);
\draw [fill=black] (3.2095745901953703,0.34917135972242674) circle (1.0pt);
\draw [fill=black] (3.439975987860061,0.11876996205773586) circle (1.0pt);
\draw [fill=black] (2.6079709407375664,0.016369340873428805) circle (1.0pt);
\draw [fill=black] (2.9407729595865644,1.19397648449296) circle (1.0pt);
\draw [fill=black] (4.121633276408737,0.5697829072153996) circle (1.0pt);
\draw [fill=black] (4.628575928217337,0.5987510587473195) circle (1.0pt);
\draw [fill=black] (4.150601427940657,0.04835617964084034) circle (1.0pt);
\draw [fill=black] (4.672028155515218,0.07732433117276029) circle (1.0pt);
\draw [fill=black] (4.3702917985587355,0.9022006773980903) circle (1.0pt);
\draw [fill=black] (3.828890553725365,0.2965653865675416) circle (1.0pt);
\draw [fill=black] (4.970848454734126,0.6162432352297802) circle (1.0pt);
\draw [fill=black] (5.017127941065096,0.10716888558911593) circle (1.0pt);
\draw [fill=black] (4.169072116515638,-0.28401252321621184) circle (1.0pt);
\draw [fill=black] (4.690905902865687,-0.27814922224598665) circle (1.0pt);
\end{scriptsize}
\end{tikzpicture}\\
Fig. 22. The hierarchy of certain classes of graphs related to co-TT graphs together with separating examples.
   
\end{figure}
 Now we are well equipped to present the main result of this paper.\\
 \begin{theo}
     A chordal graph $G$ is a co-TT graph if and only if  $G$ does not contains any graph of the infinite families $\mathcal{\widetilde{S}}, H^*_n (n\geq 2)$ or any of the graphs $T, W, H'_3$ and $T_0$ as an induced subgraph.
 \end{theo}
 \begin{proof}
     Let $G$ be a co-TT graph. From Lemmas 13, 14, 16 and 20,  $G$ does not contain $W, T, T_0$ and any graph of the family $H^*_n (n\geq 2$)  as an induced subgraph. Also we have observed that $G$ does not contain $H_3'$ and any graph of the family $\mathcal{\widetilde{S}}$ as an induced subgraph.\\
     For the converse, we assume that $G$ is a chordal graph. Now, we shall show that if $G$ is not a co-TT graph then $G$ must contain one of the graphs of the infinite families $\mathcal{\widetilde{S}}, H^*_n (n\geq 2)$ or one of the graphs $T, T_0, W, H'_3$ as an induced subgraph. We consider two possibilities:\\
     $(a)$ $G$ is a split graph. If $G$ is not a co-TT graph then $G$ contains a forbidden graph of split co-TT graphs i.e. any graph of the family $\mathcal{\widetilde{S}}$ as an induced subgraph.\\
     $(b)$ $G$ is a strongly chordal graph but not a split graph. Now such a strongly chordal graph may be an interval graph (see Fig. 22). If $G$ is an interval graph, then $G$ is not a forbidden graph for co-TT graphs. Thus $G$ must not be an interval graph. Hence by Theorem 3, $G$ contains an asteroidal triple of vertices. Then by Theorem 4, $G$ must contains any of the graphs $T$, $W$ or any graph of the families $T_n$ or $H_n$ (see Fig. 4) as an induced subgraph. Next, if $G$ is not a co-TT graph and contains any graph of the family $T_n$ then from \cref{obs17} and proof of the Lemma 16, it follows that  $G$ must contains $T$ or $T_0$ as an induced subgraph. Finally, $G$ is not a co-TT graph and contains any graph of the family $H_n$, then from \cref{obs19} and \cref{obs21} it must contains a graph of the family $H^*_n (n\geq 2)$ or $H'_3$ as an induced subgraph. This completes the proof of the theorem.
 \end{proof}
 \section{Conclusion}
In this paper, we have successfully addressed the open problem of characterizing co-TT graphs in terms of forbidden induced subgraphs originally posed by Monma, Reed, and Trotter \cite{24}. Additionally, they provided a recognition algorithm for co-TT graphs with a runtime of O($n^4$). P.A. Golovach et al. \cite{11} introduced an O($n^2$) algorithm for recognizing threshold tolerance graphs and their complements, namely co-TT graphs. We anticipate that our findings will contribute to the development of a more efficient algorithm for recognizing co-TT graphs.

\end{document}